\documentclass[a4paper]{scrartcl}
 
\usepackage{fullpage}

\usepackage{amsmath,amsthm,amssymb}
\usepackage{xspace}

\newcommand{\NP}{\ensuremath{\mathsf{NP}}\xspace}
\newcommand{\FPT}{\ensuremath{\mathsf{FPT}}\xspace}

\newcommand{\containment}{\ensuremath{\mathsf{NP \subseteq coNP/poly}}\xspace}

\newcommand{\uC}{\ensuremath{\hat{\mathcal{C}}_3}\xspace}
\newcommand{\mC}{\ensuremath{\mathcal{C}_3}\xspace}
\newcommand{\sC}{\ensuremath{\mathcal{C}_1}\xspace}
\newcommand{\usC}{\ensuremath{\hat{\mathcal{C}}_1}\xspace}
\newcommand{\msC}{\ensuremath{\mathcal{C}_1}\xspace}
\newcommand{\relC}{\ensuremath{\mathcal{C}_{\mathtt{rel}}}\xspace}
\newcommand{\irrC}{\ensuremath{\mathcal{C}_{\mathtt{irr}}}\xspace}
\newcommand{\irrM}{\ensuremath{M_{\mathtt{irr}}}\xspace}
\newcommand{\irrD}{\ensuremath{V_{\mathtt{irr}}}\xspace}

\newcommand{\fieldF}{\ensuremath{\mathbb{F}}\xspace}
\newcommand{\I}{\ensuremath{\mathcal{I}}\xspace}
\newcommand{\N}{\mathbb{N}}

\newcommand{\R}{\mathbb{R}}
\newcommand{\T}{\ensuremath{\mathcal{T}}\xspace}

\newcommand{\bX}{\ensuremath{\overline{X}}\xspace}
\newcommand{\opX}{\ensuremath{X_{\mathtt{op}}}\xspace}
\newcommand{\Y}{\ensuremath{\mathcal{Y}}\xspace}

\newcommand{\Oh}{\mathcal{O}}

\newcommand{\yes}{\textbf{yes}\xspace}
\newcommand{\no}{\textbf{no}\xspace}

\DeclareMathOperator{\opt}{opt}
\DeclareMathOperator{\tw}{tw}

\theoremstyle{plain}
\newtheorem{lemma}{Lemma}
\newtheorem{theorem}{Theorem}
\newtheorem{proposition}{Proposition}
\newtheorem{claim}{Claim}
\theoremstyle{definition}
\newtheorem{definition}{Definition}

\theoremstyle{remark}
\newtheorem{observation}{Observation}

\newcommand{\dunion}{\mathbin{\dot\cup}}

\newcommand{\probname}[1]{\lowercase{\textsc{#1}}}
\newcommand{\problem}[1]{\probname{#1}\xspace}

\newcommand{\vc}{\probname{Vertex Cover}\xspace}
\newcommand{\vck}{\probname{Vertex Cover$(k)$}\xspace}
\newcommand{\vcalp}{\probname{Vertex Cover}$(k-\mbox{\textsc{lp}})$\xspace}
\newcommand{\vcamm}{\probname{Vertex Cover}$(k-\mbox{\textsc{mm}})$\xspace}
\newcommand{\vcanb}{\probname{Vertex Cover}$(k-(2\mbox{\textsc{lp}}-\mbox{\textsc{mm}}))$\xspace}
\newcommand{\atwosat}{\probname{Almost 2-SAT($k$)}\xspace}

\newcommand{\problembox}[4]{
\begin{center}
\framebox{
    \begin{minipage}{0.95\textwidth}
    \problem{#1}\\
    \textbf{Input:} #2 \\
    \textbf{Parameter:} #3 \\
    \textbf{Question:} #4
\end{minipage}
}
\end{center}
}

\title{A randomized polynomial kernelization for Vertex Cover with a smaller parameter}

\author{Stefan Kratsch\\
Universit\"at Bonn, Institut f\"ur Informatik\\
\texttt{kratsch@cs.uni-bonn.de}}

\begin{document}

\maketitle

\begin{abstract}
In the \problem{Vertex Cover} problem we are given a graph $G=(V,E)$ and an integer $k$ and have to determine whether there is a set $X\subseteq V$ of size at most $k$ such that each edge in $E$ has at least one endpoint in $X$. The problem can be easily solved in time $\Oh^*(2^k)$, making it fixed-parameter tractable (FPT) with respect to $k$. While the fastest known algorithm takes only time $\Oh^*(1.2738^k)$, much stronger improvements have been obtained by studying \emph{parameters that are smaller than~$k$}. Apart from treewidth-related results, the arguably best algorithm for \problem{Vertex Cover} runs in time $\Oh^*(2.3146^p)$, where $p=k-LP(G)$ is only the excess of the solution size $k$ over the best fractional vertex cover (Lokshtanov et al.\ TALG 2014). Since $p\leq k$ but $k$ cannot be bounded in terms of $p$ alone, this strictly increases the range of tractable instances.

Recently, Garg and Philip (SODA 2016) greatly contributed to understanding the parameterized complexity of the \problem{Vertex Cover} problem. They prove that $2LP(G)-MM(G)$ is a lower bound for the vertex cover size of $G$, where $MM(G)$ is the size of a largest matching of $G$, and proceed to study parameter $\ell=k-(2LP(G)-MM(G))$. They give an algorithm of running time $\Oh^*(3^\ell)$, proving that \problem{Vertex Cover} is FPT in $\ell$. It can be easily observed that $\ell\leq p$ whereas $p$ cannot be bounded in terms of $\ell$ alone.
We complement the work of Garg and Philip by proving that \problem{Vertex Cover} admits a randomized polynomial kernelization in terms of $\ell$, i.e., an efficient preprocessing to size polynomial in $\ell$. This improves over parameter $p=k-LP(G)$ for which this was previously known (Kratsch and Wahlstr\"om FOCS 2012).
\end{abstract}


\section{Introduction}

A \emph{vertex cover} of a graph $G=(V,E)$ is a set $X\subseteq V$ such that each edge $e\in E$ has at least one endpoint in $X$. The \vc problem of determining whether a given graph $G$ has a vertex cover of size at most $k$ has been an important benchmark problem in parameterized complexity for both \emph{fixed-parameter tractability} and \emph{(polynomial) kernelization},\footnote{Detailed definitions can be found in Section~\ref{section:preliminaries}. Note that we use $\ell$, rather than $k$, as the default symbol for parameters and use \vc{}$(\ell)$ to refer to the \vc problem with parameter $\ell$.} which are the two notions of tractability for parameterized problems. Kernelization, in particular, formalizes the widespread notion of efficient preprocessing, allowing a rigorous study  (cf.~\cite{Kratsch14}).
We present a randomized polynomial kernelization for \vc for the to-date smallest parameter, complementing a recent fixed-parameter tractability result~by~Garg~and~Philip~\cite{GargP16}.

Let us first recall what is known for the so-called \emph{standard parameterization} \vck, i.e., with parameter $\ell=k$: There is a folklore $\Oh^*(2^k)$ time\footnote{We use $\Oh^*$ notation, which suppresses polynomial factors.} algorithm for testing whether a graph $G$ has a vertex cover of size at most $k$, proving that \vck is fixed-parameter tractable (\FPT); this has been improved several times with the fastest known algorithm due to Chen et al.~\cite{ChenKX10} running in time $\Oh^*(1.2738^k)$. Under the Exponential Time Hypothesis of Impagliazzo et al.~\cite{ImpagliazzoPZ01} there is no algorithm with runtime $\Oh^*(2^{o(k)})$. The best known kernelization for \vck reduces any instance $(G,k)$ to an equivalent instance $(G',k')$ with $|V(G')|\leq 2k$; the total size is $\Oh(k^2)$~\cite{ChenKJ01}. Unless \containment and the polynomial hierarchy collapses there is no kernelization to size $\Oh(k^{2-\varepsilon})$~\cite{DellM14}.

At first glance, the \FPT and kernelization results for \vck seem essentially best possible. This is true for parameter $\ell=k$, but there are \emph{smaller parameters} $\ell'$ for which both \FPT-algorithms and polynomial kernelizations are known. The motivation for this is that even when $\ell'=\Oh(1)$, the value $\ell=k$ may be as large as $\Omega(n)$, making both \FPT-algorithm and kernelization for parameter $k$ useless for such instances (time $2^{\Omega(n)}$ and size guarantee $\Oh(n)$). In contrast, for $\ell'=\Oh(1)$ an \FPT-algorithm with respect to $\ell'$ runs in polynomial time (with only leading constant depending on $\ell'$). Let us discuss the relevant type of smaller parameter, which relates to \emph{lower bounds on the optimum} and was introduced by Mahajan and Raman~\cite{MahajanR99}; two other types are discussed briefly under related work.

Two well-known lower bounds for the size of vertex covers for a graph $G=(V,E)$ are the maximum size of a matching of $G$ and the smallest size of fractional vertex covers for $G$; we (essentially) follow Garg and Philip~\cite{GargP16} in denoting these two values by $MM(G)$ and $LP(G)$. Note that the notation $LP(G)$ comes from the fact that fractional vertex covers come up naturally in the linear programming relaxation of the \vc problem, where we must assign each vertex a fractional value such that each edge is incident with total value of at least $1$. In this regard, it is useful to observe that the LP relaxation of the \problem{Maximum Matching} problem is exactly the dual of this. Accordingly, we have $MM(G)\leq LP(G)$ since each integral matching is also a fractional matching, i.e., with each vertex incident to a total value of at most $1$. Similarly, using $VC(G)$ to denote the minimum size of vertex covers of $G$ we get $VC(G)\geq LP(G)$ and, hence, $VC(G)\geq LP(G)\geq MM(G)$.

A number of papers have studied vertex cover with respect to ``above lower bound'' parameters $\ell'=k-MM(G)$ or $\ell''=k-LP(G)$ \cite{RazgonO09,RamanRS11,CyganPPW13,NarayanaswamyRRS12,LokshtanovNRRS14}. Observe that
\[
 k\geq k-MM(G) \geq k-LP(G).
\]
For the converse, note that $k$ can be unbounded in terms of $k-MM(G)$ and $k-LP(G)$, whereas $k-MM(G)\leq 2(k-LP(G))$ holds~\cite{KratschW12,Jansen_Thesis}. Thus, from the perspective of achieving fixed-parameter tractability (and avoiding large parameters) both parameters are equally useful for improving over parameter $k$. Razgon and O'Sullivan~\cite{RazgonO09} proved fixed-parameter tractability of \atwosat, which implies that \vcamm is \FPT due to a reduction to \atwosat by Mishra et al.~\cite{MishraRSSS11}. Using $k-MM(G)\leq 2(k-LP(G))$, this also entails fixed-parameter tractability of \vcalp. 

After several improvements~\cite{RamanRS11,CyganPPW13,NarayanaswamyRRS12,LokshtanovNRRS14} the fastest known algorithm, due to Lokshtanov et al.~\cite{LokshtanovNRRS14}, runs in time $\Oh^*(2.3146^{k-MM(G)})$. The algorithms of Narayanaswamy et al.~\cite{NarayanaswamyRRS12} and Lokshtanov et al.~\cite{LokshtanovNRRS14} achieve the same parameter dependency also for parameter $k-LP(G)$. The first (and to our knowledge only) kernelization result for these parameters is a randomized polynomial kernelization for \vcalp by Kratsch and Wahlstr\"om~\cite{KratschW12}, which of course applies also to the larger parameter $k-MM(G)$.

Recently, Garg and Philip~\cite{GargP16} made an important contribution to understanding the parameterized complexity of the \vc problem by proving it to be \FPT with respect to parameter $\ell=k-(2LP(G)-MM(G))$. Building on an observation of Lov\'asz and Plummer~\cite{LovaszP1986} they prove that $VC(G)\geq 2LP(G)-MM(G)$, i.e., that $2LP(G)-MM(G)$ is indeed a lower bound for the minimum vertex covers size of any graph $G$. They then design a branching algorithm with running time $\Oh^*(3^\ell)$ that builds on the well-known Gallai-Edmonds decomposition for maximum matchings to guide its branching choices.

\problembox{\vcanb}{A graph $G=(V,E)$ and an integer $k\in\N$.}{$\ell=k-(2LP(G)-MM(G))$ where $LP(G)$ is the minimum size of fractional vertex covers for $G$ and $MM(G)$ is the maximum cardinality of matchings of $G$.}{Does $G$ have a vertex cover of size at most $k$, i.e., a set $X\subseteq V$ of size at most $k$ such that each edge of $E$ has at least one endpoint in $X$?}

Since $LP(G)\geq MM(G)$, we clearly have $2LP(G)-MM(G)\geq LP(G)$ and hence $\ell= k-(2LP(G)-MM(G))$ is indeed at most as large as the previously best parameter $k-LP(G)$. We can easily observe that $k-LP(G)$ cannot be bounded in terms of $\ell$: For any odd cycle $C$ of length $2s+1$ we have $LP(C)=\frac12(2s+1)$, $VC(C)=s+1$, and $MM(C)=s$. Thus, a graph $G$ consisting of $t$ vertex-disjoint odd cycles of length $2s+1$ has $LP(G)=\frac12t(2s+1)$, $VC(G)=t(s+1)$, and $MM(G)=ts$. For $k=VC(G)=t(s+1)$ we get
\[
\ell=k - (2LP(G)-MM(G))=t(s+1) - t(2s+1) + ts=0
\]
whereas
\[
k-LP(G) = t(s+1) - \frac12t(2s+1) = \frac12t(2s+2) - \frac12t(2s+1) = \frac12t.
\]
Generally, it can be easily proved that $LP(G)$ and $2LP(G)-MM(G)$ differ by exactly $\frac12$ on any \emph{factor-critical} graph (cf.~Proposition~\ref{proposition:factorcritical:minvc:minfvc}). 

As always in parameterized complexity, when presented with a new fixed-parameter tractability result, the next question is whether the problem also admits a polynomial kernelization. It is well known that decidable problems are fixed-parameter tractable if and only if they admit a (not necessarily polynomial) kernelization.\footnote{We sketch this folklore fact for \vcanb: If the input is larger than $3^\ell$, where $\ell=k-(2LP(G)-MM(G))$, then the algorithm of Garg and Philip~\cite{GargP16} runs in polynomial time and we can reduce to an equivalent small yes- or no-instance; else, the instance size is bounded by $3^\ell$; in both cases we get size at most $3^\ell$ in polynomial time. The converse holds since a kernelization followed by any brute-force algorithm on an instance of, say, size $g(\ell)$ gives an \FPT running time in terms of $\ell$.} Nevertheless, not all problems admit polynomial kernelizations and, in the present case, both an extension of the methods for parameter $k-LP(G)$ \cite{KratschW12} or a lower bound proof similar to Cygan et al.~\cite{CyganLPPS14} or Jansen~\cite[Section 5.3]{Jansen_Thesis} (see related work) are conceivable.

\subparagraph{Our result.}
We give a randomized polynomial kernelization for \vcanb. This improves upon parameter $k-LP(G)$ by giving a strictly smaller parameter for which a polynomial kernelization is known. At high level, the kernelization takes the form of a (randomized) polynomial parameter transformation from \vcanb to \vcamm, i.e., a polynomial-time many-one (Karp) reduction with \emph{output parameter polynomially bounded in the input parameter}. It is well known (cf.~Bodlaender et al.~\cite{BodlaenderTY11}) that this implies a polynomial kernelization for the source problem, i.e., for \vcanb in our case. Let us give some more details of this transformation.

Since the transformation is between different parameterizations of the same problem, it suffices to handle parts of any input graph $G$ where the input parameter $\ell=k-(2LP(G)-MM(G))$ is (much) smaller than the output parameter $k-MM(G)$. After the well-known LP-based preprocessing (cf.~\cite{GargP16}), the difference in parameter values is equal to the number of vertices that are exposed (unmatched) by any maximum matching $M$ of $G$.
Consider the Gallai-Edmonds decomposition $V=A\dunion B\dunion D$ of $G=(V,E)$, where $D$ contains the vertices that are exposed by at least one maximum matching, $A=N(D)$, and $B=V\setminus (A\cup D)$. Let $M$ be a maximum matching and let $t$ be the number of exposed vertices. There are $t$ components of $G[D]$ that have exactly one exposed vertex each. The value $2LP(G)-MM(G)$ is equal to $|M|+t$ when $LP(G)=\frac12|V|$, as implied by LP-based preprocessing.

To reduce the difference in parameter values we will remove all but $\Oh(\ell^4)$ components of $G[D]$ that have an exposed vertex; they are called \emph{unmatched components} for lack of a matching edge to $A$ and we can ensure that they are not singletons. It is known that any such component $C$ is factor-critical and hence has no vertex cover smaller than $\frac12(|C|+1)$; this exactly matches its contribution to $|M|+t$: It has $\frac12(|C|-1)$ edges of $M$ and one exposed vertex. Unless the instance is trivially \no all but at most $\ell$ of these components $C$ have a vertex cover of size $\frac12(|C|+1)$, later called a \emph{tight vertex cover}. The only reason not to use a tight vertex cover for $C$ can be due to adjacent vertices in $A$ that are not selected; this happens at most $\ell$ times. A technical lemma proves that this can always be traced to at most three vertices of $C$ and hence at most three vertices in $A$ that are adjacent with $C$.

In contrast, there are (matched, non-singleton) components $C$ of $G[C]$ that together with a matched vertex $v\in A$ contribute $\frac12(|C|+1)$ to the lower bound due to containing this many matching edges. To cover them at this cost requires not selecting vertex $v$. This in turn propagates along $M$-alternating paths until the cover picks both vertices of an $M$-edge, which happens at most $\ell$ times, or until reaching an unmatched component, where it may help prevent a tight vertex cover. We translate this effect into a two-way separation problem in an auxiliary directed graph. Selecting both vertices of an $M$-edge is analogous to a adding a vertex to the separator. Relative to a separator the question becomes which sets of at most three vertices of $A$ that can prevent tight vertex covers are still reachable by propagation. At this point we can apply representative set tools from Kratsch and Wahlstr\"om~\cite{KratschW12} to identify a small family of such triplets that works for all separators (and hence for all so-called \emph{dominant} vertex covers) and keep only the corresponding components.

\subparagraph{Related work.}
Let us mention some further kernelization results for \vc with respect to nonstandard parameters. There are two further types of interesting parameters:
\begin{enumerate}
 \item \emph{Width-parameters:} Parameters such as treewidth allow dynamic programming algorithms running in time, e.g., $\Oh^*(2^{\tw})$, independently of the size of the vertex cover. It is known that there are no polynomial kernels for \vc (or most other \NP-hard problems) under such parameters~\cite{BodlaenderDFH09}. The treewidth of a graph is upper bounded by the smallest vertex cover, whereas graphs of bounded treewidth can have vertex cover size $\Omega(n)$.
 \item \emph{``Distance to tractable case''-parameters:} \vc can be efficiently solved on forests. By a simple enumeration argument it is fixed-parameter tractable when $\ell$ is the minimum number of vertices to delete such that $G$ becomes a forest. Jansen and Bodlaender~\cite{JansenB13} gave a polynomial kernelization to $\Oh(\ell^3)$ vertices. Note that the vertex cover size is an upper bound on $\ell$, whereas trees can have unbounded vertex cover size. The \FPT-result can be carried over to smaller parameters corresponding to distance from larger graph classes on which \vc is polynomial-time solvable, however, Cygan et al.~\cite{CyganLPPS14} and Jansen~\cite[Section 5.3]{Jansen_Thesis} ruled out polynomial kernels for some of them. E.g., if $\ell$ is the deletion-distance to an outerplanar graph then there is no kernelization for \vc{}$(\ell)$ to size polynomial in $\ell$ unless the polynomial hierarchy collapses~\cite{Jansen_Thesis}.
\end{enumerate}

\subparagraph{Organization.}
Section~\ref{section:preliminaries} gives some preliminaries. In Section~\ref{section:tightvertexcovers:factorcritical} we discuss vertex covers of factor-critical graphs and prove the claimed lemma about critical sets. Section~\ref{section:nicedecompositions} introduces a relaxation of the Gallai-Edmonds decomposition, called \emph{nice decomposition}, and Section~\ref{section:nicedecompositionsandvertxcovers} explores the relation between nice decompositions and vertex covers. The kernelization for \vcanb is given in Section~\ref{section:kernelization}. In Section~\ref{section:proofofmatroidresult} we provide for self-containment a result on representative sets that follows readily from~\cite{KratschW12}. We conclude in Section~\ref{section:conclusion}.


\section{Preliminaries}\label{section:preliminaries}

We use the shorthand~$[n]:=\{1,\ldots,n\}$. By $A\dunion B$ to denote the disjoint union of $A$ and $B$.

\subparagraph{Parameterized complexity.}
Let us recall that a \emph{parameterized problem} is a set $Q\subseteq\Sigma^*\times\N$ where $\Sigma$ is any finite alphabet, i.e., a language of pairs $(x,\ell)$ where the component $\ell\in\N$ is called the \emph{parameter}.
Recall also that a classical (unparameterized) problem is usually given as a set (language) $L\subseteq\Sigma^*$. For the classical problem \problem{Vertex Cover}, with instances $(G,k)$, asking whether $G$ has a vertex cover of size at most $k$, the canonical parameterized problem is \vck where the parameter value is simply $\ell=k$; this is the same procedure for any other decision problem obtained from an optimization problem by asking whether $\opt\leq k$ resp.\ $\opt\geq k$ and is called the \emph{standard parameterization}. We remark that this notation is usually abused by, e.g., using $(G,k)$ for an instance of \vck rather than the redundant $((G,k),k)$; we will use $(G,k)$ for $((G,k),k)$ and $(G,k,\ell)$ for $((G,k),\ell)$.

A parameterized problem $Q$ is \emph{fixed-parameter tractable} (\FPT) if there exists a function $f\colon\N\to\N$, a constant $c$, and an algorithm $A$ that correctly decides $(x,\ell)\in Q$ in time $f(\ell)\cdot |x|^c$ for all $(x,\ell)\in\Sigma^*\times\N$. A parameterized problem $Q$ has a \emph{kernelization} if there is a function $g\colon\N\to\N$ and a polynomial-time algorithm $K$ that on input $(x,\ell)$ returns an instance $(x',\ell')$ with $|x'|,\ell'\leq g(\ell)$ and with $(x,\ell)\in Q$ if and only if $(x',\ell')\in Q$. The function $g$ is called the \emph{size} of the kernelization $K$ and a polynomial kernelization requires that $g$ is polynomially bounded. A \emph{randomized (polynomial) kernelization} may err with some probability, in which case the returned instance is not equivalent to the input instance. Natural variants with one-side error respectively bounded error are defined completely analogous to randomized algorithms. For a more detailed introduction to parameterized complexity we recommend the recent books by Downey and Fellows~\cite{DowneyF13} and Cygan et al.~\cite{CyganFKLMPPS15}.

\subparagraph{Graphs.}
We require both directed and undirected graphs; all graphs are finite and simple, i.e., they have no parallel edges or loops. Accordingly, an undirected graph $G=(V,E)$ consists of a finite set $V$ of vertices and a set $E\subseteq\binom{V}{2}$ of edges; a directed graph $H=(V,E)$ consists of a finite set $V$ and a set $E\subseteq V^2\setminus\{(v,v)\mid v\in V\}$. For clarity, all undirected graphs are called $G$ and all directed graphs are called $H$ (possibly with indices etc.).
For a graph $G=(V,E)$ and vertex set $X\subseteq V$ we use $G-X$ to denote the graph induced by $V\setminus X$; we also use $G-v$ if $X=\{v\}$.
Analogous definitions are used for directed graphs $H$.

Let $H=(V,E)$ be a directed graph and let $S$ and $T$ be two not necessarily disjoint vertex sets in $H$. A set $X\subseteq V$ is an \emph{$S,T$-separator} if in $G-X$ there is no path from $S\setminus X$ to $T\setminus X$; note that $X$ may overlap both $S$ and $T$ and that $S\cap T\subseteq X$ is required. The set $T$ is \emph{closest to $S$} if there is no $S,T$-separator $X$ with $X\neq T$ and $|X|\leq|T|$, i.e., if $T$ is the unique minimum $S,T$-separator in $G$ (cf.~\cite{KratschW12}). Both separators and closeness have analogous definitions in undirected graphs but they are not required here.

\begin{proposition}[cf.~\cite{KratschW12}]\label{proposition:closest}
Let $H=(V,E)$ be a directed graph and let $S,T\subseteq V$ such that $T$ is closest to $S$. For any vertex $v\in V\setminus T$ that is reachable from $S$ in $H-T$ there exist $|T|+1$ (fully) vertex-disjoint paths from $S$ to $T\cup\{v\}$.
\end{proposition}

\begin{proof}
Assume for contradiction that such $|T|+1$ directed paths do not exist. By Menger's Theorem there must be an $S,T\cup\{v\}$-separator $X$ of size at most $|T|$. Observe that $X\neq T$ since $v$ is reachable from $S$ in $H-T$. Thus, $X$ is an $S,T$-separator of size at most $|T|$ that is different from $T$; this contradicts closeness of $T$.
\end{proof}

For an undirected graph $G=(V,E)$, a \emph{matching} is any set $M\subseteq E$ such that no two edges in $M$ have an endpoint in common. If $M$ is a matching in $G=(V,E)$ then we will say that a path is $M$-alternating if its edges are alternatingly from $M$ and from $\overline{M}:=E\setminus M$. An $M,M$-path is an $M$-alternating path whose first and last edge are from $M$; it must have odd length. Similarly, we define $\overline{M},M$-paths, $M,\overline{M}$-paths (both of even length), and $\overline{M},\overline{M}$-paths (of odd length). If $M$ is a matching of $G$ and $v$ is incident with an edge of $M$ then we use $M(v)$ to denote the other endpoint of that edge, i.e., the \emph{mate} or \emph{partner} of $v$. Say that a vertex $v$ is \emph{exposed by $M$} if it is not incident with an edge of $M$; we say that $v$ is \emph{exposable} if it is exposed by some maximum matching of $G$. A graph $G=(V,E)$ is \emph{factor-critical} if for each vertex $v\in V$ the graph $G-v$ has a perfect matching (a \emph{near-perfect matching of $G$}); observe that all factor-critical graphs must have an odd number of vertices.

A \emph{vertex cover} of a graph $G=(V,E)$ is a set $X\subseteq V$ such that each edge $e\in E$ has at least one endpoint in $X$. There is a well-known linear programming relaxation of the \vc problem for a graph $G=(V,E)$:
\begin{align*}
\min \quad& \sum_{v\in V} x(v)\\
s.t. \quad& x(u)+x(v)\geq 1\\
&x(v)\geq 0
\end{align*}
The optimum value of this linear program can be computed in polynomial time and it is denoted $LP(G)$. The feasible solutions $x\colon V\to\R_{\geq 0}$ are called fractional vertex covers; the \emph{cost} of a solution/fractional vertex cover $x$ is $\sum_{v\in V} x(v)$. It is well-known that the extremal points $x$ of the linear program are half-integral, i.e., $x\in\{0,\frac12,1\}^V$. With this in mind, we will tacitly assume that all considered fractional vertex covers are half-integral. We will often use the simple fact that the size of any matching $M$ of $G$ lower bounds both the cardinality of vertex covers and the cost of fractional vertex covers of $G$.

\subparagraph{Gallai-Edmonds decomposition.} We will now introduce the Gallai-Edmonds decomposition following the well-known book of Lov\'asz and Plummer~\cite{LovaszP1986}.\footnote{We use $B$ instead $C$ for $V\setminus (A\cup D)$ to leave the letter $C$ for cycles and connected components.}

\begin{definition}\label{definition:ged}
Let $G=(V,E)$ be a graph. The \emph{Gallai-Edmonds decomposition} of $G$ is a partition of $V$ into three sets $A$, $B$, and $D$ where
\begin{itemize}
 \item $D$ consists of all vertices $v$ of $G$ such that there is a maximum matching $M$ of $G$ that contains no edge incident with $v$, i.e., that leaves $v$ exposed,
 \item $A$ is the set of neighbors of $D$, i.e., $A:=N(D)$, and
 \item $B$ contains all remaining vertices, i.e., $B:=V\setminus(A\cup D)$.
\end{itemize}
\end{definition}

It is known (and easy to verify) that the Gallai-Edmonds decomposition of any graph $G$ is unique and can be computed in polynomial time. The Gallai-Edmonds decomposition has a number of useful properties; the following theorem states some of them.

\begin{theorem}[cf.\ {\cite[Theorem~3.2.1]{LovaszP1986}}]\label{theorem:ged}
Let $G=(V,E)$ be a graph and let $V=A\dunion B\dunion D$ be its Gallai-Edmonds decomposition. The following properties hold:
\begin{enumerate}
 \item The connected components of $G[D]$ are factor-critical.
 \item The graph $G[B]$ has a perfect matching.
 \item Every maximum matching $M$ of $G$ consists of a perfect matching of $G[B]$, a near-perfect matching of each component of $G[D]$, and a matching of $A$ into $D$.
\end{enumerate}
\end{theorem}


\section{Tight vertex covers of factor-critical graphs}\label{section:tightvertexcovers:factorcritical}

In this section we study vertex covers of factor-critical graphs, focusing on those that are of smallest possible size (later called tight vertex covers). We first recall the fact that any factor-critical graph with $n\geq 3$ vertices has no vertex cover of size less than $\frac12(n+1)$. By a similar argument such graphs have no fractional vertex cover of cost less than $\frac12n$.

\begin{proposition}[folklore]\label{proposition:factorcritical:minvc:minfvc}
Let $G=(V,E)$ be a factor-critical graph with at least three vertices. Every vertex cover $X$ of $G$ has cardinality at least $\frac12(|V|+1)$ and every fractional vertex cover $x\colon V\to \R_{\geq 0}$ of $G$ has cost at least $\frac12|V|$.
\end{proposition}

\begin{proof}
Let $X\subseteq V$ be a vertex cover of $G$. Since $G$ has at least three vertices and is factor-critical, it has a maximum matching $M$ of size $\frac12(|V|-1)\geq 1$. It follows that $X$ has size at least one. (This is not true for graphs consisting of a single vertex, which are also factor-critical. All other factor-critical graphs have at least three vertices.) Pick any vertex $v\in X$. Since $G$ is factor-critical, there is a maximum matching $M_v$ of $G-v$ of size $\frac12(|V|-1)$. It follows that $X$ must contain at least one vertex from each edge of $M_v$, and no vertex is contained in two of them. Together with $v$, which is not in any edge of $M_v$, this gives a lower bound of $1+\frac12(|V|-1)=\frac12(|V|+1)$, as claimed.

Let $x\colon V\to \R_{\geq 0}$ be a fractional vertex cover of $G$. We use again the matching $M$ of size at least one from the previous case; let $\{u,v\}\in M$. It follows that $x(u)+x(v)\geq 1$; w.l.o.g. we have $x(v)\geq \frac12$. Let $M_v$ be a maximum matching of $G-v$ of size $\frac12(|V|-1)$. For each edge $\{p,q\}\in M_v$ we have $x(p)+x(q)\geq 1$. Since the matching edges are disjoint we get a lower bound of $\sum_{p\in V\setminus\{v\}} x(p)\geq \frac12(|V|-1)$. Together with $x(v)\geq \frac12$ we get the claimed lower bound of $\frac12|V|$ for the cost of $x$.
\end{proof}

Note that Proposition~\ref{proposition:factorcritical:minvc:minfvc} is tight for example for all odd cycles of length at least three, all of which are factor-critical. We now define tight vertex covers and critical sets.

\begin{definition}[tight vertex covers, critical sets]
Let $G=(V,E)$ be a factor-critical graph with $|V|\geq 3$. A vertex cover $X$ of $G$ is \emph{tight} if $|X|=\frac12(|V|+1)$. Note that this is different from a minimum vertex cover, and a factor-critical graph need not have a tight vertex cover; e.g., odd cliques with at least five vertices are factor-critical but have no tight vertex cover.

A set $Z\subseteq V$ is called a \emph{bad set} of $G$ if there is no tight vertex cover of $G$ that contains $Z$. The set $Z$ is a \emph{critical set} if it is a minimal bad set, i.e., no tight vertex cover of $G$ contains $Z$ but for all proper subsets $Z'$ of $Z$ there is a tight vertex cover containing $Z'$.
\end{definition}

Observe that a factor-critical graph $G=(V,E)$ has no tight vertex cover if and only if $Z=\emptyset$ is a critical set of $G$. It may be interesting to note that a set $X\subseteq V$ of size $\frac12(|V|+1)$ is a vertex cover of $G$ if and only if it contains no critical set. (We will not use this fact and hence leave its two line proof to the reader.) The following lemma proves that all critical sets of a factor-critical graph have size at most three; this is of central importance for our kernelization. For the special case of odd cycles, the lemma has a much shorter proof and we point out that all critical sets of odd cycles have size exactly three.

\begin{lemma}\label{lemma:criticalsets:boundsize}
Let $G=(V,E)$ be a factor-critical graph with at least three vertices. All critical sets $Z$ of $G$ have size at most three.
\end{lemma}

\begin{proof}
Let $\ell\in\N$ with $\ell\geq1$ such that $|V|=2\ell+1$; recall that all factor-critical graphs have an odd number of vertices.

Assume for contradiction that there is a critical set $Z$ of $G$ of size at least four. Let $w,x,y,z\in Z$ be any four pairwise different vertices from $Z$. Let $M$ be a maximum matching of $G-w$. Since $G$ is factor-critical, we get that $M$ is a perfect matching of $G-w$ and has size $|M|=\ell$. Observe that any tight vertex cover of $G$ that contains $w$ must contain exactly one vertex from each edge of $M$, since its total size is $\frac12(|V|+1)=\ell+1$. We will first analyze $G$ and show that the presence of certain structures would imply that some proper subset $Z'$ of $Z$ is bad, contradicting the assumption that $Z$ is critical. Afterwards, we will use the absence of these structures to find a tight vertex cover that contains $Z$, contradicting the fact that it is a critical set.

If there is an $M,M$-path from $x$ to $y$ then $\{w,x,y\}$ is a bad set, i.e., no tight vertex cover of $G$ contains all three vertices $w$, $x$, and $y$, contradicting the choice of $Z$: Let $P=(v_1,v_2,\ldots,v_{p-1},v_p)$ denote an $M,M$-path from $v_1=x$ to $v_p=y$. Accordingly, we have $\{v_1,v_2\},\ldots,\{v_{p-1},v_p\}\in M$ and the path $P$ has odd length. Assume that $X$ is a tight vertex cover containing $w$, $x$, and $y$. It follows, since $w\in X$, that $X$ contains exactly one vertex per edge in $M$; in particular it contains exactly one vertex per matching edge on the path $P$. Since $v_1=x\in X$ we have $v_2\notin X$. Thus, as $\{v_2,v_3\}$ is an edge of $G$, we must have $v_3\in X$ to cover this edge; this in turn implies that $v_4\notin X$ since it already contains $v_3$ from the matching edge $\{v_3,v_4\}$. Continuing this argument along the path $P$ we conclude that $v_{p-1}\in X$ and $v_p\notin X$, contradicting the fact that $v_p=y\in X$. Thus, if there is an $M,M$-path from $x$ to $y$ then there is no tight vertex cover of $G$ that contains $w$, $x$, and $y$, making $\{w,x,y\}$ a bad set and contradicting the assumption that $Z$ is a critical set. It follows that there can be no $M,M$-path from $x$ to $y$. The same argument can be applied also to $x$ and $z$, and to $y$ and $z$, ruling out $M,M$-paths connecting them.

Similarly, if there is an edge $\{u,v\}\in M$ such that $z$ reaches both $u$ and $v$ by (different, not necessarily disjoint) $M,\overline{M}$-paths then no tight vertex cover of $G$ contains both $w$ and $z$, contradicting the choice of $Z$: Let $P=(v_1,v_2,\ldots,v_{p-1},v_p)$ denote an $M,\overline{M}$-path from $v_1=z$ to $v_p=u$ with $\{v_1,v_2\},\{v_3,v_4\},\ldots,\{v_{p-2},v_{p-1}\}\in M$. Let $X$ be a tight vertex cover of $G$ that contains $w$ and $z$. It follows (as above) that $v_1,v_3,\ldots,v_{p-2}\in X$ and $v_2,v_4,\ldots,v_{p-1}\notin X$, by considering the induced $M,M$-path from $z=v_1$ to $v_{p-1}$. The fact that $v_{p-1}\notin X$ directly implies that $v_p=u\in X$ in order to cover the edge $\{v_{p-1},v_p\}$. Repeating the same argument on an $M,\overline{M}$-path from $z$ to $v$ we get that $v\in X$. Thus, we conclude that $u$ and $v$ are both in $X$, contradicting the fact that $X$ must contain exactly one vertex of each edge in $X$. Hence, there is no tight vertex cover of $G$ that contains both $w$ and $z$. We conclude that $\{w,z\}$ is a bad set, contradicting the choice of $Z$. Hence, there is no edge $\{u,v\}\in M$ such that $z$ has $M,\overline{M}$-paths (not necessarily disjoint) to both $u$ and $v$.

Now we will complete the proof by using the established properties, i.e., the non-existence of certain $M$-alternating paths starting in $z$, to construct a tight vertex cover of $G$ that contains all of $Z$, giving the final contradiction. Using minimality of $Z$, let $X$ be a tight vertex cover of $G$ that contains $Z\setminus\{z\}$; by choice of $Z$ we have $z\notin X$. We construct the claimed vertex cover $X'\supseteq Z$ from $X'=X$ as follows:
\begin{enumerate}
 \item Add vertex $z$ to $X$ and remove $M(z)$, i.e., remove the vertex that $z$ is matched to.
 \item Add all vertices $v$ to $X'$ that can be reached from $z$ by an $M,\overline{M}$-path.
 \item Remove all vertices from $X'$ that can be reached from $z$ by an $M,M$-path of length at least three. (There is a single such path of length one from $z$ to $M(z)$ which, for clarity, was handled already above.)
\end{enumerate}
We need to check four things: (1) The procedure above is well-defined, i.e., no vertex can be reached by both $M,M$- and $M,\overline{M}$-paths from $z$. (2) The size of $X'$ is at most $|X|=\ell+1$. (3) $X'$ is a vertex cover. (4) The set $X'$ contains $w$, $x$, $y$, and $z$.

(1) Assume that there is a vertex $v$ such that $z$ reaches $v$ both by an $M,M$-path $P=(v_1,v_2,\ldots,v_p)$ with $v_1=z$ and $v_p=v$, and by an $M,\overline{M}$-path $P'$. Observe that $\{v_{p-1},v_p\}\in M$ since $P$ is an $M,M$-path and, hence, that $P''=(v_1,\ldots,v_{p-1})$ is an $M,\overline{M}$-path from $v$ to $v_{p-1}$. Together, $P'$ and $P''$ constitute two $M,\overline{M}$-paths from $z$ to both endpoints $v_{p-1}$ and $v_p$ of the matching edge $\{v_{p-1},v_p\}$; a contradiction (since we ruled out this case earlier).

(2) In the first step, we add $z$ and remove $M(z)$. Note that $z\notin X$ implies that $M(z)\in X$ (we start with $X'=X$). Thus the size of $X'$ does not change. Consider a vertex $v$ that is added in the second step, i.e., with $v\notin X$: There is an $M,\overline{M}$-path $P$ from $z$ to $v$. Since $w\in X$ we know that $v\neq w$. Thus, since $M$ is a perfect matching of $G-w$, there is a vertex $u$ with $u=M(v)$. The vertex $u:=M(v)$ must be in $X$ to cover the edge $\{v,u\}\in M$, as $v\notin X$. Moreover, $u$ cannot be on $P$ since that would make it incident with a second matching edge other than $\{u,v\}$. Thus, by extending $P$ with $\{v,u\}$ we get an $M,M$-path from $z$ to $u$, implying that $u$ is removed in the second step. Since $u\in X$ the total size change is zero. Observe that the vertex $u=M(v)$ used in this argument is not used for any other vertex $v'$ added in the second step since it is only matched to $v$. Similarly, due to (1), the vertex $u$ is not also added in the second step since it cannot be simultaneously have an $M,\overline{M}$-path from $z$.

(3) Assume for contradiction that some edge $\{u,v\}$ is not covered by $X'$, i.e., that $u,v\notin X'$. Since $w\in X'$ is the only unmatched vertex it follows that both $u$ and $v$ are incident with some edge of $M$. We distinguish two cases, namely (a) $\{u,v\}\in M$ and (b) $\{u,v\}\notin M$.

(3.a) If $\{u,v\}\in M$ then without loss of generality assume $u\in X$ (as $X$ is a vertex cover). By our assumption we have $u\notin X'$, which implies that we have removed it on account of having an $M,M$-path $P$ from $z$ to $u$. Since $\{u,v\}\in M$ the path $P$ must visit $v$ as its penultimate vertex; there is no other way for an $M,M$-path to reach $u$. This, however, implies that there is an $M,\overline{M}$-path from $z$ to $v$, and that we have added $v$ in the second step; a contradiction.

(3.b) In this case we have $\{u,v\}\notin M$. Again, without loss of generality, assume that $u\in X$. Since $u\notin X'$ there must be an $M,M$-path $P$ from $z$ to $u$. If $P$ does not contain $v$ then extending $P$ by edge $\{u,v\}\notin M$ would give an $M,\overline{M}$-path from $z$ to $v$ and imply that $v\in X'$; a contradiction. In the remaining case, the vertex $v$ is contained in $P$; let $P'$ denote the induced path from $z$ to $v$ (not containing $u$ as it is the final vertex of $P$). Since $v\notin X'$ we know that $P'$ cannot be an $M,\overline{M}$-path, or else we would have $v\in X'$, and hence it must be an $M,M$-path. Now, however, extending $P'$ via $\{v,u\}\notin M$ yields an $M,\overline{M}$-path from $z$ to $u$, contradicting (1). Altogether, we conclude that $X'$ is indeed a vertex cover.

(4) Clearly, $z\in X'$ by construction. Similarly, $w\in X'$ since it is contained in $X$ and it cannot be removed since there is no incident $M$-edge (i.e., no $M,M$-paths from $z$ can end in $w$). Finally, regarding $x$ and $y$, we proved earlier that there are no $M,M$-paths from $z$ to $x$ or from $z$ to $y$. Thus, since both $x$ and $y$ are in $X$ they must also be contained in $X'$.

We have showed that under the assumption of minimality of $Z$ and using $|Z|\geq 4$ one can construct a vertex cover $X'$ of optimal size $\ell+1$ that contains $Z$ entirely. This contradicts the choice of $Z$ and completes the proof.
\end{proof}


\section{(Nice) relaxed Gallai-Edmonds decomposition}\label{section:nicedecompositions}

The Gallai-Edmonds decomposition of a graph has a number of strong properties and, amongst others, has played a vital role in the FPT-algorithm for Garg and Philip~\cite{GargP16}. It is thus not surprising that we find it rather useful for the claimed kernelization. Unfortunately, in the context of reduction rules, there is the drawback that the Gallai-Edmonds decomposition of a graph and that the graph obtained from the reduction rule might be quite different. (E.g., even deleting entire components of $G[D]$ may ``move'' an arbitrary number of vertices from $A\cup D$ to $B$.) We cope with this problem by defining a relaxed variant of this decomposition. The relaxed form is no longer unique, but when applying certain reduction rules the created graph can effectively inherit the decomposition. 

The definition mainly drops the requirement that $D$ is the set of exposable vertices and instead allows any set $D$ that gives the desired properties. Moreover, instead of a (strong) statement about all maximum matchings of $G$ (cf.\ Definition~\ref{definition:ged}) we simply require that a single maximum matching $M$ with appropriate properties be given along with $V=A\dunion B\dunion D$.

\begin{definition}[relaxed Gallai-Edmonds decomposition]\label{definition:relaxedged}
Let $G=(V,E)$ be a graph. A \emph{relaxed Gallai-Edmonds decomposition of $G$} is a tuple $(A,B,D,M)$ where $V=A\dunion B\dunion D$ and $M$ is a maximum matching of $G$ such that
\begin{enumerate}
 \item $A=N(D)$,
 \item each connected component of $G[D]$ is factor-critical,
 \item $M$ restricted to $B$ is a perfect matching of $G[B]$,
 \item $M$ restricted to any component $C$ of $G[D]$ is a near-perfect matching of $G[C]$, and
 \item each vertex of $A$ is matched by $M$ to a vertex of $D$.
\end{enumerate}
\end{definition}

\begin{observation}
Let $G=(V,E)$ be a graph and let $(A,B,D,M)$ be a relaxed Gallai-Edmonds decomposition of $G$. For each connected component $C$ of $G[D]$ we have $N(C)\subseteq A$. (Note that this is purely a consequence of $N(D)=A$ and $C$ being a connected component of $G[D]$.)
\end{observation}

It will be of particular importance for us in what way the matching $M$ of a decomposition $(A,B,D,M)$ of $G$ matches vertices of $A$ to vertices of components of $G[D]$. We introduce appropriate definitions next. In particular, we define sets $\msC$, $\usC$, $\mC$, $\uC$, $A_1$, and $A_3$ that are derived from $(A,B,D,M)$ and $G$ in a well-defined way. Whenever we have a decomposition $(A,B,D,M)$ of $G$ we will use these sets without referring again to this definition. We will use, e.g., $\msC'$ in case where we require these sets for two decomposed graphs $G$ and $G'$.

\begin{definition}[matched/unmatched connected components of {$G[D]$}]\label{definition:mumcomponents}
Let $G=(V,E)$ be a graph and let $(A,B,D,M)$ be a relaxed Gallai-Edmonds decomposition of $G$. We say that a connected component $C$ of $G[D]$ is \emph{matched} if there are vertices $v\in C$ and $u\in N(C)\subseteq A$ such that $\{u,v\}\in M$; we will also say that $u$ and $C$ are matched to one another. Otherwise, we say that $C$ is \emph{unmatched}. Note that edges of $M$ with both ends in $C$ have no influence on whether $C$ is matched or unmatched.

We use $\msC$ and $\usC$ to denote the set of matched and unmatched singleton components in $G[D]$. We use $\mC$ and $\uC$ for matched and unmatched non-singleton components. By $A_1$ and $A_3$ we denote the set of vertices in $A$ that are matched to singleton respectively non-singleton components of $G[D]$; note that $A=A_1\dunion A_3$. We remark that the names $\mC$ and $\uC$ refer to the fact that these components have at least three vertices each as they are factor-critical and non-singleton.
\end{definition}

\begin{observation}
Let $G=(V,E)$ be a graph and let $(A,B,D,M)$ be a relaxed Gallai-Edmonds decomposition of $G$. If a component $C$ of $G[D]$ is matched then there is a unique edge of $M$ that matches a vertex of $C$ with a vertex in $A$. This is a direct consequence of $M$ inducing a near-perfect matching on $G[C]$, i.e., that only a single vertex of $C$ is not matched to another vertex of $C$.
\end{observation}

We now define the notion of a nice relaxed Gallai-Edmonds decomposition (short: nice decomposition), which only requires in addition that there are no unmatched singleton components with respect to decomposition $(A,B,D,M)$ of $G$, i.e., that $\usC=\emptyset$. Not every graph has a nice decomposition, e.g., the independent sets (edgeless graphs) have none. For the moment, we will postpone the question of how to actually find a nice decomposition (and how to ensure that there is one for each considered graph). 

\begin{definition}[nice decomposition]\label{definition:nicedecomposition}
Let $(A,B,D,M)$ be a relaxed Gallai-Edmonds decomposition of a graph $G$. We say that $(A,B,D,M)$ is a \emph{nice relaxed Gallai-Edmonds decomposition} (short a \emph{nice decomposition}) if there are no unmatched singleton components.
\end{definition}

In the following section we will derive several lemmas about how vertex covers of $G$ and a nice decomposition $(A,B,D,M)$ of $G$ interact. For the moment, we will only prove the desired property that certain operations for deriving a graph $G'$ from $G$ allow $G'$ to effectively inherit the nice decomposition of $G$ (and also keep most of the related sets $\msC$ etc.\ the same).

\begin{lemma}\label{lemma:inheritance}
Let $G=(V,E)$ be a graph, let $(A,B,D,M)$ be a relaxed Gallai-Edmonds decomposition, and let $C\in\usC\dunion\uC$ be an unmatched component of $G[D]$. Then $(A,B,D',M')$ is a relaxed Gallai-Edmonds decomposition of $G'=G-C$ where $M'$ is $M$ restricted to $V(G')=V\setminus C$ and where $D':=D\setminus C$. The corresponding sets $A_1$, $A_3$, $\msC$, and $\mC$ are the same as for $G$. The sets $\usC$ and $\uC$ differ only by the removal of component $C$, i.e., $\usC'=\usC\setminus\{C\}$ and $\uC'=\uC\setminus\{C\}$. Moreover, if $(A,B,D,M)$ is a nice decomposition then so is $(A,B,D',M')$.
\end{lemma}

\begin{proof}
Clearly, $A\dunion B\dunion D'$ is a partition of the vertex set of $G'$. Next, let us prove first that $M'$ is a maximum matching of $G'$: To get $M'$ we delete the edges of $M$ in $C$ and we delete all vertices in $C$. Thus, any matching of $G'$ that is larger than $M'$ could be extended to matching larger than $M$ for $G$ by adding the edges of $M$ on vertices of $C$. Now, we consider the connected components of $G'[D']$: We deleted the entire component $C$ of $G[D]$ to get $G'=G-C$. It follows that the connected components of $G'[D']$ are the same except for the absence of $C$, and they are factor-critical since that holds for all components of $G[D]$. Moreover, for any component $C'$ of $G[D']$ the set $M'$ induces a near-perfect matching, as it is the restriction of $M$ to $G-C$. Similarly, since $B\cap C=\emptyset$ the set $M'$ induces a perfect matching on $G[B]$. In the same way, if $\{u,v\}\in M$ where $u\in A$ and $v\in C'$ where $C'$ is a connected component of $G[D]$ other than $C$ then $u$ is also matched to a component of $G'[D']$ in $G'$, namely to $C'$. It follows that $A\subseteq N_{G'}(D')$ using that $C$ is unmatched. The inverse inclusion follows since $N_G(D)=A$ and we did not make additional vertices adjacent to $D'$. Thus, $N_{G'}(A)=D'$. Thus, $(A,B,D',M')$ is a relaxed Gallai-Edmonds decomposition.

Let us now check that the sets $A_1$, $A_3$, etc.\ are almost the same: We already saw that matching edges between vertices of $A$ and components of $G[D]$ persist in $G'$. It follows that $A'_1=A_1$ and $A'_3=A_3$, and that $\mC'=\mC$ and $\msC'=\msC$. If $C\in\usC$ then we get $\usC'=\usC\setminus\{C\}$; else we get $\usC'=\usC=\usC\setminus\{C\}$. Similarly, $\uC'=\uC\setminus\{C\}$. This completes the main statement of the lemma. The moreover part follows because $(A,B,D,M)$ being a nice decomposition of $G$ implies $\usC=\emptyset$, which yields $\usC'=\emptyset$ and, hence, that $(A,B,D',M')$ is a nice decomposition of $G'$.
\end{proof}


\section{Nice decompositions and vertex covers}\label{section:nicedecompositionsandvertxcovers}

In this section we study the relation of vertex covers $X$ of a graph $G$ and any nice decomposition $(A,B,D,M)$ of $G$. As a first step, we prove a lower bound of $|M|+|\uC|$ on the size of vertex covers of $G$; this bound holds also for relaxed Gallai-Edmonds decompositions. Additionally, we show that $|M|+|\uC|=2LP(G)-MM(G)$, if $(A,B,D,M)$ is a nice decomposition. Note that Garg and Philip~\cite{GargP16} proved that $2LP(G)-MM(G)$ is a lower bound for the vertex cover size for every graph $G$, but we require the bound of $|M|+|\uC|$ related to our decompositions, and the equality to $2LP(G)-MM(G)$ serves ``only'' to later relate to the parameter value $\ell=k-(2LP(G)-MM(G))$.

\begin{lemma}\label{lemma:nice:vclb}
Let $G=(V,E)$ be a graph and let $(A,B,D,M)$ be a nice decomposition of $G$. Each vertex cover of $G$ has size at least $|M|+|\uC|=2LP(G)-MM(G)$.
\end{lemma}

\begin{proof}
Let $X$ be any vertex cover of $G$. For each edge of $M$ that is not in a component of $|\uC|$ the set $X$ contains at least one endpoint of $M$. For each component $C\in\uC$ the set $X$ contains at least $\frac12(|C|+1)$ vertices of $C$ by Proposition~\ref{proposition:factorcritical:minvc:minfvc}, since $C$ is also a component of $G[D]$ and all those components are factor-critical. Since $M$ contains a near-perfect matching of $C$, i.e., of cardinality $\frac12(|C|-1)$, the at least $\frac12(|C|+1)$ vertices of $C$ in $X$ can also be counted as one vertex per matching edge in $G[C]$ plus one additional vertex. Overall, the set $X$ contains at least $|M|+|\uC|$ vertices.

We now prove that $|M|+|\uC|=2LP(G)-MM(G)$; first we show that $LP(G)=\frac12|V|$. Let $x\colon V\to\{0,\frac12,1\}$ be a fractional vertex cover of $V$. For each edge $\{u,v\}\in M$ that is not in a component of $\uC$ we have $x(u)+x(v)\geq 1$, in other words, $\frac12$ per vertex of the matching edge. For each $C\in \uC$ we have $\sum_{v\in C} x(v)\geq\frac12|C|$ by Proposition~\ref{proposition:factorcritical:minvc:minfvc} since $G[C]$ is factor-critical; again this equals $\frac12$ per vertex (of $C$). Using that $(A,B,D,M)$ is a nice decomposition, we can show that these considerations yield a lower bound of $\frac12(|V|)$; it suffices to check that all vertices have been considered: Components in $\msC\cup\mC$ are fully matched, all vertices in $A$ are matched (to components in $\msC\cup\mC$), and $M$ restricts to a perfect matching of $G[B]$; all these vertices contribute $\frac12$ each since they are in an edge of $M$ that is not in a component of $\uC$. All remaining vertices are in components of $\uC$ since $\usC=\emptyset$; these vertices contribute $\frac12$ per vertex by being in some component $C\in\usC$ that contributes $\frac12|C|$. Overall we get that the $x$ has cost at least $\frac12|V|$ and, hence, that $LP(G)\geq\frac12|V|$. Since $x(v)\equiv\frac12$ is a feasible fractional vertex cover for every graph, we conclude that $LP(G)=\frac12|V|$.

Now, let us consider $MM(G)$: Note that $MM(G)=|M|$ as $M$ is a maximum matching of $G$. Since $\usC=\emptyset$, we know that the only exposed vertices (w.r.t.~$M$) are in components of $\uC$; exactly one vertex per component. Thus, $|V|=2|M|+|\uC|$, which implies
\[
2LP(G)-MM(G)=|V|-|M|=2|M|+|\uC|-|M|=|M|+|\uC|.
\]
This completes the proof.
\end{proof}

Intuitively, if the size of a vertex cover $X$ is close to the lower bound of $|M|+|\uC|$ then, apart from few exceptions (at most as many as the excess over the lower bound), it contains exactly one vertex per matching edge and exactly $\frac12(|C|+1)$ vertices per component $C\in\uC$, i.e., it induces a tight vertex cover on all but few components $C\in\uC$.

Our analysis of vertex covers $X$ in relation to a fixed nice decomposition will focus on those parts of the graph where $X$ exceeds the number of one vertex per matching edge respectively $\frac12(|C|+1)$ vertices per (unmatched, non-singleton) component $C\in\uC$. To this end, we introduce the terms \emph{active component} and a set $\opX\subseteq X$, which essentially capture the places where $X$ ``overpays'', i.e., where it locally exceeds the lower bound.

\begin{definition}[active component]\label{definition:activecomponent}
Let $G=(V,E)$ be a graph, let $(A,B,D,M)$ be a nice decomposition of $G$, and let $X$ be a vertex cover of $G$. A component $C\in\uC$ is \emph{active (with respect to $X$)} if $X$ contains more than $\frac12(|C|+1)$ vertices of $C$, i.e., if $X\cap C$ is not a tight vertex cover of $G[C]$. 
\end{definition}

\begin{definition}[set $\opX$]\label{definition:setxop}
Let $G=(V,E)$ be a graph and let $(A,B,D,M)$ be a nice decomposition of $G$. For $X\subseteq V$ define \emph{$\opX=\opX(A_1,A_3,M,X)\subseteq A\cap X$} to contain all vertices $v$ that fulfill either of the following two conditions:
\begin{enumerate}
 \item $v\in A_1$ and $X$ contains both $v$ and $M(v)$.
 \item $v\in A_3$ and $X$ contains $v$.\label{definition:setxop:condition2}
\end{enumerate}
\end{definition}

Both conditions of Definition~\ref{definition:setxop} capture parts of the graph where $X$ contains more vertices than implied by the lower bound. To see this for the second condition, note that if $v\in A_3\cap X$ then $X$ still needs at least $\frac12(|C|+1)$ vertices of the component $C\in\mC$ that $v$ is matched to; since there are $\frac12(|C|+1)$ matching edges that $M$ has between vertices of $C\cup\{v\}$ we find that $X$ (locally) exceeds the lower bound, as $|X\cap(C\cup\{v\})|\geq 1+\frac12(|C|+1)$. Conversely, if $X$ does match the lower bound on $C\cup\{v\}$ then it cannot contain $v$.

We now prove formally that a vertex cover $X$ of size close to the lower bound of Lemma~\ref{lemma:nice:vclb} has only few active components and only a small set $\opX\subseteq X$.

\begin{lemma}\label{lemma:nice:boundxh:boundac}
Let $G=(V,E)$ be a graph, let $(A,B,D,M)$ be a nice decomposition of $G$, let $X$ be a vertex cover of $G$, and let $\opX=\opX(A_1,A_3,M,X)$. The set $\opX$ has size at most $\ell$ and there are at most $\ell$ active components in $\uC$ with respect to $X$ where $\ell=|X|-(|M|+|\uC|)=|X|-(2LP(G)-MM(G))$.
\end{lemma}

\begin{proof}
By Lemma~\ref{lemma:nice:vclb} we have that $X$ has size at least $|M|+|\uC|=2LP(G)-MM(G)$. Let $\ell=|X|-(|M|+|\uC|)$.

Assume first that $|\opX|>\ell$. Let $\overline{M}\subseteq M$ denote the matching edges between a vertex of $A_1$ and the vertex of a (matched) singleton component from $\msC$ that $X$ contains both endpoints of. Let $\overline{A}_3:=A_3\cap X$. By definition of $\opX$ we get that $|\opX|=|\overline{M}|+|\overline{A}_3|$. For $u\in \overline{A}_3$ consider the component $C_u\in\mC$ with $\{u,v\}\in M$ and $v\in C_u$. Observe that $C_u$ is factor-critical and has at least three vertices, which implies that $X$ needs to contain at least $\frac12(|C_u|+1)$ vertices of $C_u$ (Proposition~\ref{proposition:factorcritical:minvc:minfvc}). Note that, $M$ contains exactly $\frac12(|C_u|+1)$ matching edges between vertices of $C_u\cup\{u\}$, but $X$ contains at least $\frac12(|C_u|)+1)+1$ vertices of $C\cup\{u\}$. 

Observe that the arguments of Lemma~\ref{lemma:nice:vclb} still apply. That is, for each component $C\in\uC$ the set $X$ contains at least $\frac12(|C|+1)$ of its vertices, and for all matching edges not in such a component we know that it contains at least one of its endpoints. Summing this up as in Lemma~\ref{lemma:nice:vclb} yields the lower bound of $|M|+|\uC|=2LP(G)-MM(G)$. Now, however, for each edge of $\overline{M}$ we get an extra $+1$ in the bound, and the same is true for each vertex $u\in \overline{A}_3$ since $X$ contains at least $\frac12(|C_u|)+1)+1$ vertices of $C\cup\{u\}$, which is one more than the number of matching edges on these vertices. Thus, the size of $X$ is at least $2LP(G)-MM(G)+|\overline{M}|+|\overline{A_3}|>2LP(G)-MM(G)+\ell=|X|$; a contradiction.

Assume now that there are more than $\ell$ active components. We can apply the same accounting argument as before since $X$ needs to independently contain at least one vertex per matching edge and at least $\frac12(|C|+1)$ vertices per component $C\in\mC$. Having more than $\ell$ active components, i.e., more than $\ell$ components of $C$ where $X$ has more than $\frac12(|C|+1)$ vertices would then give a lower bound of $|X|> 2LP(G)-MM(G)+\ell=|X|$; a contradiction.
\end{proof}

The central question is of course how the different structures where $X$ exceeds the lower bound interact. We are only interested in aspects that are responsible for not allowing a tight vertex cover for any (unmatched, non-singleton) components $C\in\uC$. This happens exactly due to vertices in $A$ that are adjacent to $C$ and that are not selected by $X$. Between components of $G[B]$ and non-singleton components of $G[D]$ there are $M$-alternating paths with vertices alternatingly from $A$ and from singleton components of $G[D]$ since vertices in $A$ are all matched to $D$ and singleton components in $G[D]$ have all their neighbors in $A$. Unless $X$ contains both vertices of a matching edge, it contains the $A$- or the $D$-vertices of such a path. Unmatched components of $G[D]$ and components of $G[B]$ have all neighbors in $A$. Matched components $C$ in $G[D]$ with matched neighbor $v\in A$ enforce not selecting $v$ for $X$ unless $X$ spends more than the lower bound; in this way, they lead to selection of $D$-vertices on $M$-alternating paths. Intuitively, this leads to two ``factions'' that favor either $A$- or $D$-vertices and that are effectively separated when $X$ selects both $A$- and $D$-endpoint of a matching edge. An optimal solution need not separate all neighbors in $A$ of any component $C\in\uC$, and $C$ may still have a tight vertex cover or paying for a larger cover of $C$ is overall beneficial.
The following auxiliary directed graph $H$ captures this situation and for certain vertex covers $X$ reachability of $v\in A$ in $H-\opX$ will be proved to be equivalent with $v\notin X$.

\begin{definition}[auxiliary directed graph $H$]\label{definition:graphh}
Let $G=(V,E)$ be a graph and let $(A,B,D,M)$ be a nice decomposition of $G$. Define a \emph{directed graph $H=H(G,A,B,D,M)$} on vertex set $A$ by letting $(u,v)$ be a directed edge of $H$, for $u,v\in A$, whenever there is a vertex $w\in D$ with $\{u,w\}\in E\setminus M$ and $\{w,v\}\in M$.
\end{definition}

The first relation between $G$, with decomposition~$(A,B,D,M)$, and the corresponding directed graph $H=H(G,A,B,D,M)$ is straightforward: It shows how inclusion and exclusion of vertices in a vertex cover work along an $M$-alternating path, when $X$ contains exactly one vertex per edge. We will later prove a natural complement of this lemma, but it involves significantly more work and does not hold for all vertex covers.

\begin{lemma}\label{lemma:reachable}
Let $G=(V,E)$ be a graph, let $(A,B,D,M)$ be a nice decomposition of $G$, and let $X$ be a vertex cover of $G$. Let $H=H(G,A,B,D,M)$ and $\opX=\opX(A_1,A_3,M,X)$. If $v\in A$ is reachable from $A_3$ in $H-\opX$ then $X$ does not contain $v$.
\end{lemma}

\begin{proof}
Let $P_H=(v_1,\ldots,v_p)$ be a directed path in $H-\opX$ from some vertex $v_1\in A_3\subseteq A$ to $v_p=v\in A$, and with $v_1,\ldots,v_p\in A\setminus \opX=V(H-\opX)$. By construction of $H$, for each edge $(v_i,v_{i+1})$ with $i\in[p-1]$ there is a vertex $u_i\in D$ with $\{v_i,u_i\}\in E\setminus M$ and $\{u_i,v_{i+1}\}\in M$. Since $M$ is a matching, all vertices $u_i$ are pairwise different and none of them are in $P_H$ as $u_i\in D$ and $A\cap D=\emptyset$. It follows that there is a path
\[
P=(v_1,u_1,v_2,u_2,v_3,\ldots,v_{p-1},u_{p-1},v_p)
\]
in $G$ where $\{v_i,u_i\}\in E\setminus M$ and $\{u_i,v_{i+1}\}\in M$ for $i\in[p-1]$. In other words, $P$ is an $\overline{M},M$-path from $v_1\in A_3$ to $v_p=v\in A$. 

Consider any edge $\{u_i,v_{i+1}\}\in M$ of $P$ and apply Definition~\ref{definition:setxop}: If $v_{i+1}\in A_3$ then $v_{i+1}\notin \opX$ implies that $v_{i+1}\notin X$. If $v_{i+1}\in A_1$ then $v_{i+1}\notin \opX$ implies that $X$ does not contain both $u_i$ and $v_{i+1}$. In both cases $X$ does not contain both vertices of the edge $\{u_i,v_{i+1}\}\in M$. Thus, $X$ contains exactly one vertex each from $\{u_1,v_2\},\ldots,\{u_{p-1},v_p\}$.

Let us check that this implies that $u_{p-1}\in X$ and $v_p\notin X$. Observe that $v_1\notin X$ since $v_1\in A_3$ and $v_1\in X$ would imply $v_1\in \opX$. Clearly, $X$ must then contain $u_1$ to cover the edge $\{v_1,u_1\}$, but then it does not contain $v_2$, which would be a second vertex from $\{u_1,v_2\}$. Thus, to cover $\{v_2,u_2\}$ the set $X$ must contain $u_2$, implying that it does not contain also $v_3$ from $\{u_2,v_3\}\in M$. By iterating this argument we get that $u_{p-1}\in X$ and $v_p\notin X$. Since $v=v_p$, this completes the proof.
\end{proof}

We will now work towards a complement of Lemma~\ref{lemma:reachable}: We would like to show that, under the same setup as in Lemma~\ref{lemma:reachable}, if $v$ is not reachable from $A_3$ in $H-\opX$ then $X$ does contain $v$. In general, this does not hold. Nevertheless, one can show that there always \emph{exists} a vertex cover of at most the same size, and with same set $\opX$, that does contain $v$. Equivalently, we may put further restrictions on $X$ under which the lemma holds; to this end, we define the notion of a \emph{dominant} vertex cover.

\begin{definition}[dominant vertex cover]
Let $G=(V,E)$ be a graph and let $(A,B,D,M)$ be a nice decomposition of $G$. A vertex cover $X\subseteq V$ of $G$ is \emph{dominant} if $G$ has no vertex cover of size less than $|X|$ and no vertex cover of size $|X|$ contains fewer vertices of $D$.
\end{definition}

We continue with a technical lemma that will be used to prove two lemmas about dominant vertex covers. The lemma statement is unfortunately somewhat opaque, but essentially it comes down to a fairly strong replacement routine that, e.g., can turn a given vertex cover into one that contains further vertices of $A$ and strictly less vertices of $D$.

\begin{lemma}\label{lemma:unify}
Let $G=(V,E)$ be a graph, let $(A,B,D,M)$ be a nice decomposition of $G$, and let $H=H(G,A,B,D,M)$. Let $X\subseteq V$ and $\opX=\opX(A_1,A_3,M,X)$. Suppose that there is a nonempty set $Z\subseteq A\setminus X$ such that
\begin{enumerate}
 \item $X\cup Z$ is a vertex cover of $G$,
 \item $X$ contains $M(z)$ for all $z\in Z$, and
 \item $Z$ is not reachable from $A_3$ in $H-\opX$.
\end{enumerate}
Then there exists a vertex cover $\bX$ of size at most $|X|$ that contains $Z$. Moreover, $\bX\cap D\subsetneq X\cap D$ and $\bX\cap A\supsetneq X\cap A$.
\end{lemma}

\begin{proof}
We give a proof by minimum counterexample. Assume that the lemma does not hold, and pick sets $X$ and $Z$ that fulfill the conditions of the lemma but for which the claimed set $\overline{X}$ does not exist, and with minimum value of $|X\cap D|$ among such pairs of sets. (It is no coincidence that the choice of $X$ is reminiscent of a dominant vertex cover, but note that $X$ is not necessarily a vertex cover.) We will derive sets $X'$ and $Z'$ such that either $Z'=\emptyset$ and we can choose $\bX:=X'$, or $Z'\neq\emptyset$ but then $X'$ and $Z'$ fulfill the conditions of the lemma and we have $|X'\cap D|<|X\cap D|$. In the latter case, the lemma must hold for $X'$ and $Z'$ and we will see that $\bX:=\bX'$ fulfills the then-part of the lemma for $X$ and $Z$. Thus, both cases contradict the assumption that $X$ and $Z$ constitute a counterexample, proving correctness of the lemma.

First, let us find an appropriate set $X'$. To this end, let $U:=\{M(z) \mid z\in Z\}$. Since $Z\subseteq A$ we know that $U\subseteq D$ and that each vertex $z\in Z$ is matched to a private vertex $M(z)\in U$; hence $|U|=|Z|\geq 1$. We have $U\subseteq X$ since $X$ contains $M(z)$ for all $z\in Z$. Define $X':=(X\setminus U)\cup Z$. We have $|X'|=|X|-|U|+|Z|=|X|$ since $U\subseteq X$ and $|Z|=|U|$, and because $Z\subseteq A\setminus X$ entails that $X\cap Z=\emptyset$. Moreover, since $\emptyset\neq U\subseteq D$ and $Z\cap D=\emptyset$, we get that $X'\cap D\subsetneq X\cap D$; this also means that $|X'\cap D|<|X\cap D|$. Similarly, since $\emptyset\neq Z\subseteq A$, $X\cap Z=\emptyset$, and $U\cap A=\emptyset$, we get $X'\cap A\supsetneq X\cap A$. Finally, note that $X'\cup U$ is a vertex cover since $X'\cup U = X\cup Z$ is a vertex cover.

Second, we define $Z':=\{v\mid v\in N(u) \mbox{ for some } u\in U\}\setminus X'$, i.e., $Z'$ contains all vertices $v$ that are neighbor of some $u\in U$ and that are not in $X'$. (Note that this not the same as $N(U)\setminus X'$ since a vertex $u\in U$ could have a neighbor $u'\in U$. Nevertheless, we show in a moment that $Z'\subseteq A$, ruling out this case as $U\subseteq D$ and $A\cap D=\emptyset$.) Clearly $X'\cap Z'=\emptyset$, and $Z\cap Z'=\emptyset$ since $Z\subseteq X'$. Observe that $X'\cup Z'$ is a vertex cover since $X'\cup U$ is a vertex cover: The only edges not covered by $X'\subsetneq X'\cup U$ have one endpoint in $U$ and the other one not in $X'$; these edges are covered by $Z'$ by definition.

Let us prove that $Z'\subseteq A\setminus X'$; it remains to prove $Z'\subseteq A$: Since $U\subseteq D$ and $N(D)=A$ we know that $N(u)\subseteq A\cup D$ for $u\in U$. Assume for contradiction that some $u\in U\subseteq D$ has a neighbor $v\in D$, and let $z\in Z$ with $u=M(z)$, using the definition of $U$. It follows that $u$ and $v$ are contained in the same non-singleton component $C$ of $G[D]$, as they are adjacent vertices of $D$. Moreover, $C$ is matched to $z$ since $u=M(z)$ implies $\{u,z\}\in M$. This in turn implies that $C$ is a matched non-singleton component, i.e., $C\in\mC$, and, hence, $z\in A_3$. We also know find that $z\notin\opX$ since $Z\subseteq A\setminus X$ entails $z\notin X$ (cf.\ Definition~\ref{definition:setxop}). Together, however, this implies that $z$ is reachable from $A_3$ in $H-\opX$, namely from $z\in A_3$; a contradiction. Thus, no vertex $u\in U$ has a neighbor $v\in D$, implying that $Z'\subseteq A$. Together with $X'\cap Z'=\emptyset$ we get $Z'\subseteq A\setminus X'$.

We now prove that $M(z')\in X'$ for all $z'\in Z'$. Pick any $z'\in Z'$ and note that $z'\in A\setminus X'$. Thus, $z'$ is matched to some vertex $w\in D$, i.e., $w=M(z')$. The set $X'\cup U$ is a vertex cover, implying that it contains at least one vertex of the edge $\{z',w\}$. Since $z'\in A\setminus X'\subseteq A$, it is neither in $X'$ nor in $U$ (recall that $U\subseteq D$ and $A\cap D=\emptyset$). Thus, $w\in X'\cup U$. If $w\in U$ then there exists $z\in Z$ with $w=M(z)$ by definition of $U$. Clearly, as $M$ is matching, we must have $z=z'$. This, however, violates our earlier observation that $Z\cap Z'=\emptyset$ since both sets would contain $z$. Thus, the only remaining possibility is that $w\in X'$. Hence, we get $M(z')=w\in X'$, as claimed.

Define $\opX':=\opX(A_1,A_3,M,X')$; to prove that no vertex of $Z'$ is reachable from $A_3$ in $H-\opX'$ it will be convenient to first prove $\opX\subseteq\opX'$: Let $v\in\opX$ and recall that $\opX\subseteq A$. If $v\in A_3$ then $v\in\opX$ implies that $v \in X$. By definition of $X'$ we have $v\in X'$ as only vertices in $U\subseteq D$ are in $X$ but not in $X'$. From $v\in X'$, for $v\in A_3$, we directly conclude that $v\in \opX'$. If $v\in A_1$ then $v\in\opX$ implies that $v,M(v)\in X$. This implies $v\in X'$ as before but we still need to show that $M(v)\in X'$. Assume for contradiction that $M(v)\notin X'$. Observe that this implies $M(v)\in U$ by definition of $X'$, as $M(v)\in X$. Thus, by definition of $U$, we get that $M(v)$ is matched to some vertex $z\in Z$, i.e., $M(v)=M(z)$. Since $M$ is a matching and $M(v)$ is matched to $v$, we of course get $v=z$. This implies $v=z\in Z$, which contradicts $v\in X$ as $Z\subseteq A\setminus X$. Thus, we have both $v\in X'$ and $M(v)\in X'$, which, for $v\in A_1$, implies that $v\in\opX'$. Both cases together imply that $\opX\subseteq\opX'$.

We will now prove that no vertex of $Z'$ is reachable from $A_3$ in $H-\opX'$, using $\opX\subseteq\opX'$. Let $P=(v_1,\ldots,v_p)$ be any directed path in $H$ with $v_1\in A_3$ and $v_p=z'\in Z'$.  As $z'\in Z'$ there is $u\in U$ with $z'\in N(u)\setminus X'$. Similarly, since $u\in U$ there must be $z\in Z$ with $u=M(z)$; we have $z\neq z'$ since $Z\cap Z'=\emptyset$. Observe that this means that $\{z,u\}\in M$ and $\{u,z'\}\in E\setminus M$ as $u$ cannot be incident with two matching edges. This implies, by Definition~\ref{definition:graphh}, that $(z',z)$ is an edge in $H$. Thus, there is a directed walk $W$ from $v_1\in A_3$ to $z\in Z$ in $H$ by using path $P$ and appending the edge $(z',z)$. (With slightly more work one could see that this must be a path, but we do not need this fact.) Since no vertex of $Z$ is reachable from $A_3$ in $H-\opX$ we conclude that $W$ contains at least one vertex of $\opX$. Note that $\opX$ does not contain $z\in Z$ since we assumed $Z\subseteq A\setminus X$ and $\opX\subseteq X$. Thus, $\opX$ contains a vertex of $P$ (noting that $z$ is the only vertex of $W$ that may not be in $P$). Since $\opX\subseteq\opX'$ it follows that $\opX'$ also contains a vertex of $P$; since $P$ was chosen arbitrarily it follows that no vertex of $Z'$ is reachable from $A_3$ in $H-\opX'$, as claimed.

Finally, we distinguish two cases: (1) $Z'=\emptyset$ and (2) $Z'\neq\emptyset$. In the former case, we show that $\bX:=X'$ is feasible; in the latter case we use the lemma on $X'$ and $Z'$ to get $\bX'$ and then show that $\bX:=\bX'$ fulfills the then-part of the lemma.

(1) $Z'=\emptyset$: We get that $X'=X'\cup Z'$ is a vertex cover of $G$. We showed that $|X'|=|X|$, and that $X'\cap D\subsetneq X\cap D$ and $X'\cap A\supsetneq X\cap A$. Finally, by construction we have that $Z\subseteq X'$. Thus, $\bX:=X'$ fulfills the properties claimed in the lemma, contradicting the fact that $X$ and $Z$ constitute a counterexample.

(2) $Z'\neq\emptyset$: Together with $Z'\neq\emptyset$ the above considerations show that $X'$ and $Z'$ fulfill the conditions of the lemma: The set $Z'$ is a nonempty subset of $A\setminus X'$; the set $X'\cup Z'$ is a vertex cover of $G$; the set $X'$ contains $M(z')$ for all $z'\in Z'$; and $Z'$ is not reachable from $A_3$ in $H-\opX'$, where $\opX'=\opX(A_1,A_3,M,X')$. Moreover, we know that $|X'\cap D|<|X\cap D|$, which implies that the lemma must hold for this choice of sets, as $X$ and $Z$ was assumed to be a counterexample with minimum value of $|X\cap D|$. Let $\bX'$ be the outcome of applying the lemma to $X'$ and $Z'$; let us check that $\bX:=\bX'$ is feasible:
\begin{itemize}
 \item The lemma guarantees that $\bX'$ is a vertex cover of $G$.
 \item The lemma guarantees $|\bX'|\leq |X'|$, and using $|X'|=|X|$ we conclude that $|\bX'|\leq |X|$.
 \item We know, as discussed in case (1), that $Z\subseteq X'$. The lemma guarantees that $\bX'\cap A\supsetneq X'\cap A$ and $\bX'\cap D\subsetneq X'\cap D$. The former, together with $Z\subseteq X'$ and $Z\subseteq A$, yields $Z\subseteq X'\cap A\subsetneq \bX'\cap A$. Together with $X'\cap A\supsetneq X\cap A$ and $X'\cap D\subsetneq X\cap D$, we get $\bX'\cap A\supsetneq X'\cap A \supsetneq X\cap A$ and $\bX'\cap D\subsetneq X'\cap D\subsetneq X\cap D$.
\end{itemize}
Thus, $\bX:=\bX'$ is a feasible choice. Altogether, we find that in both cases there does in fact exist a valid set $\bX$. This means that $X$ and $Z$ do not constitute a counterexample. Since there is no minimum counterexample, the lemma holds as claimed.
\end{proof}

Now, as a first application of Lemma~\ref{lemma:unify} we prove a complement to Lemma~\ref{lemma:reachable}. Note that this lemma only applies to dominant vertex covers, whereas Lemma~\ref{lemma:reachable} holds for any vertex cover of $G$. Fortunately, after the rather long proof of Lemma~\ref{lemma:unify}, the present lemma is now a rather straightforward conclusion.

\begin{lemma}\label{lemma:notreachable}
Let $G=(V,E)$ be a graph, let $(A,B,D,M)$ be a nice decomposition of $G$, and let $H=H(G,A,B,D,M)$. Let $X$ be a dominant vertex cover of $G$ and let $\opX=\opX(A_1,A_3,M,X)$. If $v\in A$ is not reachable from $A_3$ in $H-\opX$ then $X$ contains $v$.
\end{lemma}

\begin{proof}
First, let us note that if $v\in\opX$ then, by Definition~\ref{definition:setxop}, we know that $v\in X$. It remains to consider the more interesting case that $v\in A\setminus \opX$.

Assume for contradiction that $v\notin X$. We will apply Lemma~\ref{lemma:unify} to reach a contradiction. To this end, we will define a set $Z$ such that $X$ and $Z$ fulfill the conditions of Lemma~r\ref{lemma:unify}. Let $Z:=\{v\}$. Clearly, we have $\emptyset\neq Z\subseteq A\setminus X$. Since $X$ is a vertex cover and $v\notin X$, the vertex $M(v)$ must be in $X$ in order to cover the edge $\{v,M(v)\}$. (Note that $v\in A$ implies that $M(v)\in D$ exists.) By assumption of the present lemma, $v$ is not reachable from $A_3$ in $H-\opX$, where $\opX=\opX(A_1,A_3,M,X)$. Thus, Lemma~\ref{lemma:unify} applies to $X$ and $Z$, and yields a set $\bX$ that is a vertex cover of $G$ of size at most $X$ and with $|\bX\cap D|<|X\cap D|$, contradicting the assumption that $X$ is a dominant vertex cover. Thus, the assumption that $v\notin X$ is wrong, and the lemma follows.
\end{proof}

As a second application of Lemma~\ref{lemma:unify} we prove that sets $\opX$ corresponding to dominant vertex covers are always closest to $A_3$ in the auxiliary directed graph $H$. This is a requirement for applying the matroid tools from Kratsch and Wahlstr\"om~\cite{KratschW12} later since closest sets allow to translate between reachability with respect to a closest cut and independence in an appropriate matroid. Unlike the previous lemma, there is still quite some work involved before applying Lemma~\ref{lemma:unify} in the proof.

\begin{lemma}\label{lemma:closest}
Let $G=(V,E)$ be a graph, let $(A,B,D,M)$ be a nice decomposition of $G$, and let $H=H(G,A,B,D,M)$. Let $X$ be a dominant vertex cover of $G$ and let $\opX=\opX(A_1,A_3,M,X)$. Then $\opX$ is closest to $A_3$ in $H$.
\end{lemma}

\begin{proof}
Assume that $\opX$ is not closest to $A_3$ in $H$ and, consequently, let $Y\subseteq V(H)=A$ be a minimum $A_3,\opX$-separator in $H$ with $|Y|\leq|\opX|$ and $Y\neq \opX$. We will apply Lemma~\ref{lemma:unify} to appropriately chosen sets $X'$ and $Z$ (with $X'$ and $Z$ playing the roles of $X$ and $Z$ in the lemma).

Let $X':=(X\setminus (\opX\setminus Y))\cup (Y\setminus \opX)$. Note that
\begin{align*}
|\opX\setminus Y|=|\opX|-|\opX\cap Y|\geq |Y|-|\opX\cap Y|=|Y\setminus \opX|.
\end{align*}
This implies that $|X'|\leq|X|$, using that $\opX\setminus Y \subseteq \opX \subseteq X$ (see Definition~\ref{definition:setxop}). We can also observe that $X'$ and $X$ contain the same vertices of $D$, and hence also the same number since $\opX\setminus Y$ and $Y\setminus \opX$ are both subsets of $A$. (Let us mention that these two properties are not needed to apply Lemma~\ref{lemma:unify} to $X'$ but they are needed for the outcome to have relevance for $X$.) 

Let $\opX'=\opX(A_1,A_3,M,X')$ according to Definition~\ref{definition:setxop}. We show that $Y\subseteq \opX'$ by proving that $y\in\opX'$ for all $y\in Y$; we distinguish two cases depending on whether $y\in\opX$.

Let $y\in Y\cap \opX$. If $y\in A_1$ then $y\in\opX$ implies $y,M(y)\in X$. By definition of $X'$ we also have $y,M(y)\in X'$: Only elements of $\opX\setminus Y\subseteq A$ are in $X$ but not in $X'$; neither $y\in Y\cap\opX$ nor $M(y)\in D$ are affected by this. Thus, if $y\in A_1$, then $y,M(y)\in X'$, which implies $y\in \opX'$. If $y\in A_3$ then $y\in\opX$ implies $y\in X$. As before, the definition of $X'$ implies $y\in X'$, which yields $y\in \opX'$. Thus, all $y\in Y\cap\opX$ are also contained in $\opX'$.

Now, let $y\in Y\setminus\opX$. Since $Y$ is a minimal $A_3,\opX$-separator, there must be an $A_3,y$-path in $H-(Y\setminus\{y\})$ or else $Y\setminus\{y\}$ would also be an $A_3,\opX$-separator. (This is a standard argument, if $Y\setminus\{y\}$ were not a separator then there would be an $A_3,\opX$-path avoiding $Y\setminus\{y\}$. This path needs to contain $y$, as $Y$ is a separator, and can be shortened to a path from $A_3$ to $y$.) Let $P$ be a directed path from some vertex $v\in A_3$ to $y$ in $H-(Y\setminus\{y\})$, i.e., a path in $H$ containing no vertex of $Y\setminus\{y\}$. We find that there can be no vertex of $\opX$ on $P$: We already know that the final vertex $y$ of $P$ is not in $\opX$. If $u$ is any earlier vertex of $P$ that is in $\opX$ then $P$ could be shortened to a path from $v\in A_3$ to $u\in \opX$ that avoids all vertices of $Y$ (since $y$ was the only vertex of $Y$ on $P$ but it comes after $u$); thus $Y$ would not separate $A_3$ from $\opX$ in $H$. Since $P$ contains no vertex of $\opX$, we conclude that $y$ is reachable from $A_3$ in $H-\opX$. By Lemma~\ref{lemma:reachable} we conclude that $y\notin X$. Since $Y\subseteq A$, the vertex $y$ is matched to some vertex $u\in D$, and $X$ must contain $u$ to cover the edge $\{u,y\}$. Since $X$ and $X'$ contain the same vertices of $D$, as observed above, we have $u\in X'$. Additionally, by construction of $X'$, we have $Y\setminus\opX\subseteq X'$, implying that $y\in X'$. Thus, if $y\in A_1$ then we have $y\in X'$ and $M(y)=u\in X'$, which implies $y\in\opX'$; if $y\in A_3$ then $y\in X'$ suffices to conclude $y\in\opX'$. Together we get that $y\in Y\setminus\opX$ implies $y\in\opX'$; combined with the case $y\in Y\cap\opX$ we get $Y\subseteq\opX'$. 

Let $Z:=\opX\setminus Y$. By definition of $X'$ we have $X'\cap Z=\emptyset$; since $\opX\subseteq A$ this entails $Z\subseteq A\setminus X'$. Since $|Y|\leq|\opX|$ and $Y\neq\opX$, we conclude that $Z=\opX\setminus Y\neq\emptyset$. The set $X'\cup Z$ contains $X$ by definition of $X'$ and hence it is also a vertex cover of $G$. To get that $M(z)\in X'$ for $z\in Z$ we need to distinguish two cases: If $z\in A_1$ then $z\in\opX$ implies $M(z)\in X$; note that $M(z)\in D$ as $z\in A$. Since $X'$ contains the same vertices of $D$ as $X$ we get $M(z)\in X'$. If $z\in A_3$ then we reach a contradiction: Recall that $Y$ is an $A_3,\opX$-separator. This necessitates that $Y$ contains all vertices of $A_3\cap\opX$, implying that $z\in Y$, contradicting $z\in\opX\setminus Y$. Thus, if $z\in\opX\setminus Y$ then $z\in A_1$ and we get $M(z)\in X'$ as claimed. Finally, let us check that no vertex of $Z$ is reachable from $A_3$ in $H-\opX'$. This follows immediately from $Z\subseteq \opX$ and $Y\subseteq\opX'$, and the fact that $Y$ is an $A_3,\opX$-separator in $H$.

By the above considerations we may apply Lemma~\ref{lemma:unify} to $X'$ and $Z$ and obtain a vertex cover \bX of $G$ of size at most $|X'|\leq|X|$ that contains fewer vertices of $D$ than $X'$. Since $X$ and $X'$ contain the same number of vertices of $D$, we get $|\bX\cap D|<|X\cap D|$, contradicting the choice of $X$ as a dominant vertex cover.
\end{proof}


\section{Randomized polynomial kernelization}\label{section:kernelization}

In this section, we describe our randomized polynomial kernelization for \vcanb. For convenience, let us fix an input instance $(G,k,\ell)$, i.e., $G=(V,E)$ is a graph for which we want to know whether it has a vertex cover of size at most $k$; the parameter is $\ell=k-(2LP(G)-MM(G))$, where $LP(G)$ is the minimum cost of a fractional vertex cover of $G$ and $MM(G)$ is the size of a largest matching.

From previous work of Garg and Philip~\cite{GargP16} we know that the well-known linear program-based preprocessing for \vc (cf.~\cite{CyganFKLMPPS15}) can also be applied to \vcanb; the crucial new aspect is that this operation does not increase the value $k-(2LP-MM)$. The LP-based preprocessing builds on the half-integrality of fractional vertex covers and a result of Nemhauser and Trotter~\cite{NemhauserT1975} stating that all vertices with value $1$ and $0$ in an optimal fractional vertex cover $x\colon V\to\{0,\frac12,1\}$ are included respectively excluded in at least one minimum (integral) vertex cover. Thus, only vertices with value $x(v)=\frac12$ remain and the best LP solution costs exactly $\frac12$ times number of (remaining) vertices. For our kernelization we only require the fact that if $G$ is reduced under this reduction rule then $LP(G)=\frac12(|V(G)|)$; e.g., we do not require $x\colon V\to\{\frac12\}$ to be the unique optimal fractional vertex cover. Without loss of generality, we assume that our given graph $G=(V,E)$ already fulfills $LP(G)=\frac12|V|$.

\begin{observation}
If $LP(G)=\frac12|V|$ then $2LP(G)-MM(G)=|V|-MM(G)$. In other words, if $M$ is a maximum matching of $G$ then the lower bound $2LP(G)-MM(G)=|V|-MM(G)=|V|-|M|$ is equal to cardinality of $M$ plus the number of isolated vertices.
\end{observation}

As a first step, let us compute the Gallai-Edmonds decomposition $V=A\dunion B\dunion D$ of $G$ according to Definition~\ref{definition:ged}; this can be done in polynomial time.\footnote{The main expenditure is finding the set $D$. A straightforward approach is to compute a maximum matching $M_v$ of $G-v$ for each $v\in V$. If $|M_v|=MM(G)$ then $v$ is in $D$ as $M_v$ is maximum and exposes $v$; otherwise $v\notin D$ as no maximum matching exposes $v$.} Using $LP(G)=\frac12|V|$ we can find a maximum matching $M$ of $G$ such that $(A,B,D,M)$ is a nice decomposition of $G$.

\begin{lemma}\label{lemma:kernel:nicedecomposition}
Given $G=(V,E)$ with $LP(G)=\frac12|V|$ and a Gallai-Edmonds decomposition $V=A\dunion B\dunion D$ of $G$ one can in polynomial time compute a maximum matching $M$ of $G$ such that $(A,B,D,M)$ is a nice decomposition of $G$.
\end{lemma}

\begin{proof}
Let $\sC$ denote the set of singleton components of $G[D]$ and let $I=V(\sC)\subseteq D$ contain all vertices that are in singleton components of $G[D]$. Clearly, $I$ is an independent set since $G[I]$ is the subgraph of $G[D]$ containing just the singleton components. Assume for contradiction that there is a set $I'\subseteq I$ with $|N_G(I')|<|I'|$. It follows directly that there would be a fractional vertex cover of $G$ of cost less than $\frac12|V|$, namely assign $0$ to vertices of $I'$, assign $1$ to vertices of $N(I')$, and assign $\frac12$ to all other vertices. The total cost is
\begin{align*}
0\cdot |I'| + 1\cdot |N(I')| + \frac12 |V\setminus(I'\cup N(I'))|<\frac12|I'|+\frac12 |N(I')|+\frac12|V\setminus(I'\cup N(I'))| = \frac12|V|.
\end{align*}
All edges incident with $I'$ have their other endpoint in $N(I')$, which has value $1$. All other edges have two endpoints with value at least $\frac12$. This contradicts the assumption that $LP(G)=\frac12|V|$.

Thus, each $I'\subseteq I$ has at least $|I'|$ neighbors in $G$. By Hall's Theorem there exists a matching of $I'$ into $N(I')$, and standard bipartite matching algorithms can find one in polynomial time; let $M_1$ be such a matching. Using any matching algorithm that finds a maximum matching by processing augmenting paths, we can compute from $M_1$ in polynomial time a maximum matching $M$ of $G$. The matching $M$ still contains edges incident with all vertices of $I$ since extending a matching along an augmenting path does not expose any previously matched vertices.

Using the maximum matching $M$, let us check briefly that $(A,B,D,M)$ is indeed a nice decomposition of $G$. We know already that there are no unmatched singleton components since $M$ contains matching incident with all vertices of $I=V(\sC)$ and all these edges are also incident to a vertex in $A$. (Recall that the neighborhood of each component of $G[D]$ in $G$ lies in $A$.) Since $V=A\dunion B\dunion D$ is a Gallai-Edmonds decomposition of $G$ we get from Definition~\ref{definition:ged} and Theorem~\ref{theorem:ged} that $A=N(D)$, each component of $G[D]$ is factor-critical, and that $M$ (by being a maximum matching of $G$) must induce a perfect matching of $G[B]$, a near-perfect matching of each component $C$ of $G[D]$, and a matching of $A$ into $D$. This completes the proof.
\end{proof}

We fix a nice decomposition $(A,B,D,M)$ of $G$ obtained via Lemma~\ref{lemma:kernel:nicedecomposition}. We have already learned about the relation of dominant vertex covers $X$, their intersection with the set $A$, and separation of $A$ vertices from $A_3$ in $H-\opX$, where $H=H(G,A,B,D,M)$. It is safe to assume that solutions are dominant vertex covers as among minimum vertex covers there is a minimum intersection with $D$. We would now like to establish that most components of $\uC$ can be deleted (while reducing $k$ by the cost for corresponding tight vertex covers). Clearly, since any vertex cover pays at least for tight covers of these components, we cannot turn a yes- into a no-instance this way. However, if the instance is no then it might become yes.

In the following, we will try to motivate both the selection process for components of $\uC$ that are deleted as well as the high-level proof strategy for establishing correctness. We will tacitly ignore most technical details, like parameter values, getting appropriate nice decompositions, etc., and refer to the formal proof instead. Assume that we are holding a no-instance $(G,k,\ell)$. Consider for the moment, the effect of deleting all components $C\in\uC$ that have tight vertex covers and updating the budget accordingly; for simplicity, say they all have such vertex covers. Let $(G_0,k_0,\ell)$ be the obtained instance; if this instance is no as well, then deleting any subset of $\uC$ also preserves the correct answer (namely: no). Else, if $(G_0,k_0,\ell)$ is yes then pick any dominant vertex cover $X^0$ for it. We could attempt to construct a vertex cover of $G$ of size at most $k$ by adding back the components of $C$ and picking a tight vertex cover for each; crucially, these covers must also handle edges between $C$ and $A$. Since $(G,k,\ell)$ was assumed to be a no-instance, there must be too many components $C\in\uC$ for which this approach fails. For any such component, the adjacent vertices in $A\setminus X^0$ force a selection of their neighbors $Z_A=N(A)\cap C$ that cannot be completed to a tight vertex cover of $C$. To avoid turning the no-instance $(G,k,\ell)$ into a yes-instance $(G',k',\ell)$ we have to keep enough components of $\uC$ in order to falsify any suggested solution $X'$ of size at most $k'$ for $G$. The crux is that there may be an exponential number of such solutions and that we do not know any of them. This is where the auxiliary directed graph and related technical lemmas as well as the matroid-based tools of Kratsch and Wahlstr\"om~\cite{KratschW12} are essential.

Let us outline how we arrive at an application of the matroid-based tools. Crucially, if $C$ (as above) has no tight vertex cover containing $Z_A=N(A)\cap C$ then, by Lemma~\ref{lemma:criticalsets:boundsize}, there is a set $Z\subseteq Z_A$ of size at most three such that no tight vertex cover contains $Z$. Accordingly, there is a set $T\subseteq A\setminus X^0$ of size at most three whose neighborhood in $C$ contains $Z$. Thus, the fact that $X^0$ contains no vertex of $T$ is responsible for not allowing a tight vertex cover of $C$. This in turn, by Lemma~\ref{lemma:notreachable} means that all vertices in $T$ are reachable from $A_3$ in $H-\opX^0$. Recalling that a set $\opX^0$ corresponding to a dominant vertex cover is also closest to $A_3$, we can apply a result from~\cite{KratschW12} that generates a sufficiently small representative set of sets $T$ corresponding to components of $\uC$. If a dominant vertex cover has any reachable sets $T$ then the lemma below guarantees that at least one such set is in the output. For each set we select a corresponding component $C\in\uC$ and then start over on the remaining components. After $\ell+1$ iterations we can prove that for any not selected component $C$, which we delete, and any proposed solution $X'$ for the resulting graph that does not allow a tight vertex cover for $C$, there are $\ell+1$ other selected components on which $X'$ cannot be tight. This is a contradiction as there are at most $\ell$ such active components by Lemma~\ref{lemma:nice:boundxh:boundac}.

Concretely, we will use the following lemma about representative sets of vertex sets of size at most three regarding reachability in a directed graph (modulo deleting a small set of vertices). Notation of the lemma is adapted to the present application. The original result is for pairs of vertices in a directed graph (see~\cite[Lemma 2]{KratschW11_arxiv}) but extends straightforwardly to sets of fixed size $q$ and to sets of size at most $q$; a proof is provided in Section~\ref{section:proofofmatroidresult} for completeness. Note that the lemma is purely about reachability of small sets in a directed graph (like the \problem{Digraph Pair Cut} problem studied in~\cite{KratschW11_arxiv,KratschW12}) and we require the structural lemmas proved so far to negotiate between this an \vcanb.

\begin{lemma}\label{lemma:repsetofcriticalsets}
Let $H=(V_H,E_H)$ be a directed graph, let $S_H\subseteq V_H$, let $\ell\in\N$, and let $\T$ be a family of nonempty vertex sets $T\subseteq V_H$ each of size at most three. In randomized polynomial time, with failure probability exponentially small in the input size, we can find a set $\T^*\subseteq\T$ of size $\Oh(\ell^3)$ such that for any set $X_H\subseteq V_H$ of size at most $\ell$ that is closest to $S_H$ if there is a set $T\in\T$ such that all vertices $v\in T$ are reachable from $S_H$ in $H-X_H$ then there is a corresponding set $T^*\in\T^*$ satisfying the same properties.
\end{lemma}

Using the lemma we will be able to identify a small set $\relC$ of components of $\uC$ that contains for each dominant vertex cover $X$ of $G$ of size at most $k$ all active components with respect to $X$. Conversely, if there is no solution of size $k$, we will have retained enough components of $\uC$ to preserve this fact. Concretely, the set $\relC$ is computed as follows:
\begin{enumerate}
 \item Let $\relC^0$ contain all components $C\in\uC$ that have no vertex cover of size at most $\frac12(|C|+1)$. Clearly, these components are active for every vertex cover of $G$. We know from Lemma~\ref{lemma:nice:boundxh:boundac} that there are at most $\ell$ such components if the instance is \yes. We can use the algorithm of Garg and Philip~\cite{GargP16} to test in polynomial time whether any $C\in\uC$ has a vertex cover of size at most $k_C:=\frac12(|C|+1)$: We have parameter value 
 \[
 k_C-(2LP(G[C])-MM(G[C]))=\frac12(|C|+1)-(|C|-\frac12(|C|-1))=0.
 \]
 We could of course also use an algorithm for \vc parameterized above maximum matching size, where we would have parameter value $1$. If there are more than $\ell$ components $C$ with no vertex cover of size $\frac12(|C|+1)$ then we can safely reject the instance. Else, as indicated above, let $\relC^0$ contain all these components and continue.
 \item Let $i=1$. We will repeat the following steps for $i\in\{1,\ldots,\ell+1\}$.
 \item Let $\T^i$ contain all nonempty sets $T\subseteq A$ of size at most three such that there is a component $C\in\uC\setminus(\relC^0\cup\ldots\cup\relC^{i-1})$ such that:\label{step:selectti}
 \begin{enumerate}
  \item There is a set $Z\subseteq N_G(T)\cap C$ of at most three neighbors of $T$ in $C$ such that no vertex cover of $G[C]$ of size $\frac12(|C|+1)$ contains $Z$. Note that $Z\neq\emptyset$ since $C\notin \relC^0$ implies that it has at least some vertex cover of size $\frac12(|C|+1)$.
  \item For each $C$ and $Z\subseteq C$ of size at most three, existence of a vertex cover of $G[C]$ of size $k_C:=\frac12(|C|+1)$ containing $Z$ can be tested by the algorithm of Garg and Philip~\cite{GargP16} since the parameter value is constant. Concretely, run the algorithm on $G[C\setminus Z]$ and solution size $k_C-|Z|$ and observe that the parameter value is
  \[
   (k_C-|Z|)-(2LP(G[C\setminus Z])-MM(G[C\setminus Z])).
  \]
  Using that $LP(G[C\setminus Z])\geq LP(G[C])-|Z|$ and $MM(G[C\setminus Z])\leq MM(G[C])=\frac12(|C|-1)$ this value can be upper bounded by
  \begin{align*}
   & k_C-|Z|-2LP(G[C]) +2|Z| + MM(G[C])\\
   ={} & \frac12(|C|+1) - |Z| - |C| + 2 |Z| + \frac12(|C|-1)\\
   ={} & |Z|.  
  \end{align*}
  Since $|Z|\leq 3$ the parameter value is at most three and the FPT-algorithm of Garg and Philip~\cite{GargP16} runs in polynomial time.
 \end{enumerate}
 Intuitively, the condition is that $C$ must always be active for vertex covers not containing $T$, but for the formal correctness proof that we give later the above description is more convenient.
 \item Apply Lemma~\ref{lemma:repsetofcriticalsets} to graph $H=H(G,A,B,D,M)$ on vertex set $V_H=A$, set $S_H=A_3\subseteq A$, integer $\ell$, and family $\T^i$ of nonempty subsets of $A$ of size at most three to compute a subset $\T^{i*}$ of $\T^i$ in randomized polynomial time. The size of $|\T^{i*}|$ is $\Oh(\ell^3)$.
 \item Select a set $\relC^i$ as follows: For each $T\in\T^{i*}$ add to $\relC^i$ a component $C\in\uC\setminus(\relC^0\cup\ldots\cup\relC^{i-1})$ such that $C$ fulfills the condition for $T$ in Step~\ref{step:selectti}, i.e., such that: \label{step:selectci}
 \begin{enumerate}
  \item There is a set $Z\subseteq N_G(T)\cap C$ of at most three neighbors of $T$ in $C$ such that no vertex cover of $G[C]$ of size $\frac12(|C|+1)$ contains $Z$. (We know that $Z$ must be nonempty.)
 \end{enumerate}
 Clearly, the size of $|\relC^i|$ is $\Oh(\ell^3)$. Note that the same component $C$ can be chosen for multiple sets $T\in\T^{i*}$ but we only require an upper bound on $|\relC^i|$
 \item If $i<\ell+1$ then increase $i$ by one and return to Step~\ref{step:selectti}. Else return the set 
 \[
 \relC:=\bigcup_{i=0}^{\ell+1}\relC^i.
 \]
 The size of $\relC$ is $\Oh(\ell^4)$ since it is the union of $\ell+2$ sets that are each of size $\Oh(\ell^3)$.\label{step:returnrelc}
\end{enumerate}

In particular, we will be interested in the components $C\in\uC$ that are not in $\relC$. We call these \emph{irrelevant components} and let $\irrC:=\uC\setminus\relC$ denote the set of all irrelevant components. (Of course we still need to prove that they are true to their name.)

\begin{lemma}\label{lemma:removeirrelevantcomponents}
Let $G'$ be obtained by deleting from $G$ all vertices of irrelevant components, i.e., $G':=G-\bigcup_{C\in\irrC}C$, and let $k'=k-\sum_{C\in\irrC}\frac12(|C|+1)$, i.e., $k'$ is equal to $k$ minus the lower bounds for vertex covers of the irrelevant components. Then $G$ has a vertex cover of size at most $k$ if and only if $G'$ has a vertex cover of size at most $k'$. 
Moreover, $k-(2LP(G)-MM(G))=k'-(2LP(G')-MM(G'))$, i.e., the instances $(G,k,\ell)$ and $(G',k,\ell')$ of \vcanb have the same parameter value $\ell=\ell'$.
\end{lemma}

\begin{proof}
Let us first discuss the easy direction: Assume that $G$ has a vertex cover $X$ of size at most $k$; prove that $G'$ has a vertex cover of size at most $k'$. Let $X'$ denote the restriction of $X$ to $G'$, i.e., $X'=X\cap V(G')=X\setminus U$ where $U=\bigcup_{C\in\irrC} C$. Clearly, $X'$ is a vertex cover of $G'$. Concerning the size of $X'$ let us observe the following: For each component $C\in\irrC$ the set $X\cap C$ must be a vertex cover of $G[C]$ (this of course holds for any set of vertices in $G$). We know that each graph $G[C]$ for $C\in\irrC\subseteq\uC$ is factor-critical and, hence, the size of $X\cap C$ is at least $\frac12(|C|+1)$. Summing over all $C\in\irrC$ we find that $X'$ contains at least $\sum_{C\in\irrC}\frac12(|C|+1)$ vertices less than $X$. This directly implies that $|X'|\leq k'$ and completes this part of the proof.

Now, assume that $G'$ has a vertex cover of size at most $k'$; let $V':=V(G')$ and $E':=E(G')$. This part requires most of the lemmas that we established in the previous sections. It is of particular importance, that from the nice decomposition $(A,B,D,M)$ of $G$ we can derive a very similar nice decomposition of $G'$. For convenience let $\irrD:=\bigcup_{C\in\irrC}C$. By Lemma~\ref{lemma:inheritance} we may repeatedly delete unmatched components, such as $\irrC\subseteq \uC$, and always derive a nice decomposition of the resulting graph. Doing this for all components in $\irrC$ we end up with graph $G'$ and the nice decomposition $(A,B,D',M')$ where $M'$ is the restriction of $M$ to $V(G')=V\setminus \irrD$ and $D'=D\setminus \irrD$.

Let us now fix an arbitrary dominant vertex cover $X'$ of $G'$ with respect to $(A,B,D',M')$, i.e., $X'$ is of minimum size and contains the fewest vertices of $D'$ among minimum vertex covers of $G'$; clearly $|X'|\leq k'$. Our strategy will be to construct a vertex cover of $G$ of size at most $k$ by adding a vertex cover of size $\frac12(|C|+1)$ for each component $C\in\irrC$. The crux with this idea lies in the edges between components $C$ and the set $A$. We will need to show that we can cover edges between $C$ and $A\setminus X'$ by the selection of vertices in $C$ without spending more than $\frac12(|C|+1)$. Define $H':=H(G',A,B,D',M')$ according to Definition~\ref{definition:graphh} and define $\opX':=\opX(A_1,A_3,M',X')$ according to Definition~\ref{definition:setxop}; by Lemma~\ref{lemma:closest} the set $\opX'$ is closest to $A_3$ in $H'$ and by Lemma~\ref{lemma:nice:boundxh:boundac} we have $|\opX'|\leq \ell$. We claim that $H'$ is in fact identical with $H=H(G,A,B,D,M)$; let us see why this holds: Both graphs are on the same vertex set $A$. There is a directed edge $(u,v)$ in $H$ if there is a vertex $w\in D$ with $\{u,w\}\in E\setminus M$ and $\{w,v\}\in M$. Note that this implies $w\notin \irrD$ as all components in $\irrC\subseteq\uC$ are unmatched, whereas $\{w,v\}\in M$ and $w\in D$ and $v\in A$. Thus, $w$ exists also in $G'$ and $\{w,v\}\in M'$ since it is not an edge between vertices of a component in $\uC$. Similarly, $\{u,w\}\in E'\setminus M'$ since $M'\subseteq M$ and $E'$ contains all edges of $E$ that have no endpoint in $\irrD$. Thus, $(u,v)$ is also a edge of $H'$. Conversely, if $(u,v)$ is a directed edge of $H'$ then there exists $w\in D'$ with $\{w,v\}\in M'$ and $\{u,w\}\in E'\setminus M'$. Clearly, $w\in D\supseteq D'$ and $\{w,v\}\in M\supseteq M'$. Since $M$ is a matching, it cannot contain both $\{u,w\}$ and $\{w,v\}$, hence $\{u,w\}\notin M$. Thus, using $E'\subseteq E$ we have $\{u,w\}\in E\setminus M$, implying that $(u,w)$ is also an edge of $H$. Thus, the two graphs $H$ and $H'$ are identical and, in particular, $\opX'$ is also closest to $A_3$ in $H$.

Consider now the set $X'$ as a partial vertex cover of $G$. There are uncovered edges, i.e., with no endpoint in $X'$, inside components $C\in\irrC$ and between such components and vertices in $A\setminus X'$. Since the remaining budget of $k-k'$ is exactly equal to smallest vertex covers for the components in $\irrC$ we cannot add any vertices that are not in such a component (and not more than $\frac12(|C|+1)$ per component $C$). Thus, if $C\in\irrC$ and $A\setminus X'$ has a set $Z$ of neighbors in $C$, then the question is whether there is size $\frac12(|C|+1)$ vertex cover of $G$ that includes $Z$. We will prove that this is always the case.

\begin{claim}\label{claim:key}
Let $C\in\irrC\subseteq\uC$ and let $Z_C:=N_G(A\setminus X')\cap C$. There is a vertex cover $X_C$ of $G[C]$ with $Z_C\subseteq X_C$ of size at most $\frac12(|C|+1)$.
\end{claim}

\begin{proof}
Assume for contradiction that there is no vertex cover of $G[C]$ that includes $Z_C$ and has size at most $\frac12(|C|+1)$. By Lemma~\ref{lemma:criticalsets:boundsize} there is a subset $Z\subseteq Z_C$ of size at most three such that no vertex cover of $G[C]$ of size at most $\frac12(|C|+1)$ contains $Z$: Let $Z$ be any minimal subset of $Z_C$ with this property; the lemma implies that $|Z|\leq 3$. (Note that $C\in\irrC\subseteq \uC$ implies that $C$ has at least three vertices and that it is factor-critical as a component of $G[D]$.) Let $A_C$ be a minimal subset of $A\setminus X'$ such that its neighborhood in $C$ includes the set $Z$; since $Z$ has size at most three, the set $A_C$ also has size at most three. Since $A_C\cap X'=\emptyset$, by Lemma~\ref{lemma:notreachable}, each $v\in A_C$ is reachable from $A_3$ in $H'-\opX'=H-\opX'$.

We first prove that $C$ must have been considered in all $\ell+1$ iterations of computing $\relC$. If $Z=\emptyset$ then $C$ has no vertex cover of size $\frac12(|C|+1)$. This, however, would imply that $C\in\relC^0\subseteq\relC$; a contradiction. For the remainder of the proof we have $Z\neq\emptyset$, i.e., $1\leq |Z|\leq 3$, and hence the set $A_C$ must be nonempty to ensure $Z\subseteq N_G(A_C)\cap C$ (and of size at most three). It follows that in each repetition of Step~\ref{step:selectti} the sets $T=A_C\subseteq A$, component $C$, and set $Z$ were considered. (Note that $C\in\irrC=\uC\setminus \relC=\uC\setminus (\relC^0\cup\relC^1\cup\ldots\cup\relC^{\ell+1})$.) We have $T=A_C\subseteq A$ nonempty and of size at most three, $Z\subseteq N_G(T)\cap C$ of size at most three, and there is no vertex cover of $G[C]$ of size $\frac12(|C|+1)$ that contains $Z$. Thus, the set $A_C$ is contained in all sets $\T^1,\ldots,\T^{\ell+1}$.

Now, for each $i\in\{1,\ldots,\ell+1\}$, we need to consider why $C$ was not added to $\relC^i$. Let us consider two cases, namely $A_C\in\T^{i*}$ and $A_C\notin\T^{i*}$: If $A_C\in\T^{i*}$ then in Step~\ref{step:selectci} we have selected a component $C^i\neq C$, with $C^i\in\uC\setminus(\relC^0\cup\ldots\cup\relC^{i-1})$, such that there is a nonempty set $Z^i\subseteq N_G(A_C)\cap C^i$ such that $G[C^i]$ has no vertex cover of size $\frac12(|C^i|+1)$ that contains $Z^i$. For later reference let us remember the triple $(C^i,Z^i,A^i)$ with $A^i:=A_C$. We know that $C^i\in\relC^i\subseteq\relC$, we know that there is no vertex cover of $G[C^i]$ of size $\frac12(|C^i|+1)$ that contains $Z^i$, and $Z^i\subseteq N_G(A^i)\cap C^i$. Crucially, all vertices $v\in A^i$ are reachable from $A_3$ in $H-\opX'$. (There will be a second source of such triples in the case that $A_C\notin\T^{i*}$, but with $A^i\neq A_C$ and with slightly more work for proving these properties of the triples in question.)

In the second case we have $A_C\notin\T^{i*}$. By Lemma~\ref{lemma:repsetofcriticalsets}, since $A_C\in\T^i$, all vertices of $A_C$ are reachable from $A_3$ in $H-\opX'$, and $\opX'\subseteq A$ of size at most $\ell$, it follows that $\T^{i*}$ contains a set $A^i$ such that all vertices of $A^i$ are reachable from $A_3$ in $H-\opX'$. Thus, in Step~\ref{step:selectci} we have selected a component $C^i$ such that there is a set $Z^i\subseteq N_G(A_i)\cap C^i$ such that $G[C^i]$ has no vertex cover of size $\frac12(|C^i|+1)$ that contains $Z^i$. We remember the triple $(C^i,Z^i,A^i)$.

We find that, independently of whether $A_C\in\T^{i*}$ in iteration $i\in\{1,\ldots,\ell+1\}$ we get a triple $(C^i,Z^i,A_i)$ such that $C^i\in\relC^i$, with $Z^i\subseteq N_G(A_i)\cap C^i$, and such that all vertices $v\in A_i$ are reachable from $A_3$ in $H-\opX'$. We observe that the components $C^i$ are pairwise distinct: Say $1\leq i<j\leq\ell+1$. Then $C^i\in\relC^i$ and $C^j\in \uC\setminus(\relC^0\cup\ldots\cup\relC^{j-1})$, implying that $C^j\notin \relC^i$ as $i\leq j-1$, and hence that $C^j\neq C_i$. We use these components to prove that $X'$, as a vertex cover of $G'$, has at least $\ell+1$ active components, namely $C^1,\ldots,C^{\ell+1}$, which will be seen to contradict that it has size at most $k'$.

Let $i\in \{1,\ldots,\ell+1\}$. We have that all vertices of $A^i$ are reachable from $A_3$ in $H-\opX'$; the same is true in $H'-{\opX}$ since $H=H'$. From Lemma~\ref{lemma:reachable} applied to graph $G'$, nice decomposition $(A,B,D',M')$ of $G'$, and vertex cover $X'$ we get that $X'$ contains no vertex of $A^i$.
Since $C^i\in\relC$ we know that $C^i$ is also an unmatched non-singleton component of $G[D']$ with respect to matching $M'$ (Lemma~\ref{lemma:inheritance}). As $X'\cap A^i=\emptyset$ it follows directly that $X'$ contains $N_{G'}(A^i)\cap C^i=N_G(A^i)\cap C^i\supseteq Z^i$ (regarding $N_{G'}(A^i)\cap C^i=N_G(A^i)\cap C^i$ note that $G'$ differs from $G$ only by removing vertices of some other components of $\uC$, none of which are in these sets). Since $G'[C^i]=G[C^i]$ has no vertex cover of size $\frac12(|C^i|+1)$ that contains $Z^i$, it follows that $X'$ contains more than $\frac12(|C^i|+1)$ vertices of $C^i$, making $C^i$ an active component with respect to $X'$.

We proved that the vertex cover $X'$ of $G'$ has at least $\ell+1$ active components. Thus, by Lemma~\ref{lemma:nice:boundxh:boundac} its size is at least $|M'|+|\uC'|+\ell+1$ where $\uC'$ is the number of unmatched non-singleton components of $G'[D']$. On the other hand, we have $|X'|\leq k'= k-\sum_{C\in\irrC}\frac12(|C|+1)$ and $\ell=k-(2LP(G)-MM(G))=k-(|M|+|\uC|)$. It remains to compare $|M'|+|\uC'|$ with $|M|+|\uC|$.
\begin{enumerate}
 \item Let $\irrM\subseteq M$ denote the set of edges in $M$ whose endpoints are in $\irrD$; recall that $\irrD$ denotes the set of all vertices of (irrelevant) components in $\irrC$.
 \item Thus, when creating $G'$, we are deleting $|\irrD|$ vertices and $|\irrM|$ matching edges (recalling that there are no matching edges with exactly one endpoint in $\irrD$). Each component $C\in\irrC$ contributes $|C|$ vertices to $\irrD$ and $\frac12(|C|-1)$ edges to $\irrM$. Thus, $|\irrM|=\sum_{C\in\irrC}\frac12(|C|-1)=\sum_{C\in\irrC}\frac12(|C|+1)-|\irrC|$.
 \item Using this, we get
 \begin{align*}
  k'&= k-\sum_{C\in\irrC}\frac12(|C|+1)\\
    &= |M| + |\uC| + \ell - (|\irrM| + |\irrC|)\\
    &= |M'| + |\uC'| + \ell.
 \end{align*}
\end{enumerate}
Since $|X'|\geq |M'|+|\uC'|+\ell+1=k'+1$, this contradicts the assumption that $X'\leq k'$. Thus, the initial assumption in the claim proof must be wrong and we get that there does exist a vertex cover of $G[C]$ that there does exist a vertex cover $X_C$ of size $\frac12(|C|+1)$ that contains $Z_C$.
\end{proof}

Using the claim we can now easily complete $X'$ to a vertex cover $X$ of $G$ of size at most $k$: As observed before, we need to add vertices such as to cover all edges inside components $C\in\irrC$ and edges between such components and $A\setminus X'$. Begin with $X:=X'$. Consider any component $C\in\irrC$ and let $Z_C:=N_G(A\setminus X')\cap C$. By Claim~\ref{claim:key}, we know that there exists a vertex cover $X_C$ of size $\frac12(|C|+1)$ that contains $Z_C$. Clearly, by adding $X_C$ to $X$ we cover all edges of $C$ and all edges between $C$ and neighbors of $C$ that were not covered by $X'$. (The endpoints of these edges in $C$ exactly constitute the set $Z_C\subseteq X_C$.) By performing this step for all components $C\in\irrC$ we add exactly $\sum_{C\in\irrC}\frac12(|C|+1)$ vertices, implying that 
\[
|X|= |X'|+\sum_{C\in\irrC}\frac12(|C|+1) \leq k' + \sum_{C\in\irrC}\frac12(|C|+1).
\]
Since $\irrD=\bigcup_{C\in\irrC} C$ are the only vertices present in $G$ but not in $G'=G-\irrD$ it follows that all edges with no endpoint in $\irrD$ are already covered by $X'$. Thus, $X$ is indeed a vertex cover of $G$ of size at most $k$, as claimed.

It remains, to prove that $k-(2LP(G)-MM(G))=k'-(2LP(G')-MM(G'))$. We already proved that $k'=|M'|+|\uC'|+\ell$. By Lemma~\ref{lemma:nice:vclb}, since $(A,B,D',M')$ is a nice decomposition of $G'$, we have $|M'|+|\uC'|=2LP(G')-MM(G')$. This directly implies that $k'-(2LP(G')-MM(G'))=k'-(|M'|+|\uC'|)=\ell=k-(2LP(G)-MM(G))$.
\end{proof}

We can now complete our kernelization. According to Lemma~\ref{lemma:removeirrelevantcomponents} it is safe to delete all irrelevant components (and update $k$ accordingly). We obtain a graph $G'$ and integer $k'$ such that the following holds:
\begin{enumerate}
 \item $G'$ has a vertex cover of size at most $k'$ if and only if $G$ has a vertex cover of size at most $k$, i.e., the instances $(G,k)$ and $(G',k')$ for \vc are equivalent.
 \item As a part of the proof of Lemma~\ref{lemma:removeirrelevantcomponents} we showed that
 \[
  k'=|M'|+|\uC'|+\ell
 \]
 where $\uC'$ is the set of unmatched non-singleton components of $G'[D']$ with respect to $M'$.
 \item From Lemma~\ref{lemma:inheritance} we know that $\uC'$ is equal to the set $\uC$ (of unmatched non-singleton components of $G[D]$ with respect to $M$) minus the components $C\in\irrC$ that were removed to obtain $G'$. In other words, $\uC'=\uC\setminus\irrC=\relC$.
 \item We know from Step~\ref{step:returnrelc} that $|\relC|=\Oh(\ell^4)$. Hence, $|\uC'|=\Oh(\ell^4)$.
 \item Let us consider $p:=k'-|M'|$, which is the parameter value of $(G',k')$ when considered as an instance of \vc parameterized above the size of a maximum matching. Clearly,
 \[
  p=k'-|M'|=|M'|+|\uC'|+\ell-|M'|=\ell+\Oh(\ell^4)=\Oh(\ell^4).
 \]
 \item We can now apply any polynomial kernelization for \vcamm to get a polynomial kernelization for \vcanb. On input of $(G',k',p)$ it returns an equivalent instance $(G^*,k^*,p^*)$ of size $\Oh(p^c)$ for some constant $c$. We may assume that $k^*=\Oh(p^c)$ since else it would exceed the number of vertices in $G^*$ and we may as well return a \yes-instance of constant size. 
 
 Let $\ell^*=k^*-(2LP(G^*)-MM(G^*))$, i.e., the parameter value of the instance $(G^*,k^*,\ell^*)$ of \vcanb. Clearly, $\ell^*\leq k^*=\Oh(p^c)$. Thus, $(G^*,k^*,\ell^*)$ has size and parameter value $\Oh(p^c)$.
\end{enumerate}

Kratsch and Wahlstr\"om~\cite{KratschW12} give a randomized polynomial kernelization for \vcamm. The size is not analyzed since it relies on equivalence of \atwosat and \vcamm under polynomial parameter transformations~\cite{RamanRS11}; the reductions preserve the parameter value but may increase the size polynomially. The size obtained for \atwosat is $\Oh(p^{12})$, following from a kernelization to $\Oh(p^6)$ variables. Even without a size increase by the transformation back to \vcamm, which seems doable, we only get a size of $\Oh(p^{12})=\Oh(\ell^{48})$. We note, however, that the kernelization for \atwosat also relies, amongst others, on computing a representative set of reachable tuples in a directed graph. It is likely that a direct approach for kernelizing \vcanb could make do with only a single iteration of this strategy.

\begin{theorem}\label{theorem:main}
\vcanb has a randomized polynomial kernelization with error probability exponentially small in the input size.
\end{theorem}


\section{Proof of Lemma~\ref{lemma:repsetofcriticalsets}}\label{section:proofofmatroidresult}

In this section we provide a proof of Lemma~\ref{lemma:repsetofcriticalsets}, which is a generalization of \cite[Lemma 2]{KratschW11_arxiv}; in that work, it is already pointed out that a generalization to $q$-tuples is possible by the same approach. Accordingly, the proof in this section is provided only to make the present work self-contained.

We need to begin with some basics on matroids; for a detailed introduction to matroids see Oxley~\cite{OxleyBook}: A \emph{matroid} is a pair $M=(U,\I)$ where $U$ is the \emph{ground set} and $\I\subseteq 2^U$ is a family of \emph{independent sets} such that
\begin{enumerate}
 \item $\emptyset\in\I$,
 \item if $I\subseteq I'$ and $I'\in\I$ then $I\in\I$, and
 \item if $I,I'\in\I$ with $|I|<|I'|$ then there exists $u\in I'\setminus I$ with $I\cup\{u\}\in\I$; this is called the augmentation axiom.
\end{enumerate}
A set $I\in\I$ is \emph{independent}; all other subsets of $U$ are \emph{dependent}. The maximal independent sets are called \emph{bases}; by the augmentation axiom they all have the same size. For $X\subseteq U$, the \emph{rank $r(X)$ of $X$} is the cardinality of the largest independent set $I\subseteq X$. The \emph{rank of $M$} is $r(M):=r(U)$.

Let $A$ be a matrix over a field $\fieldF$, let $U$ be the set of columns of $A$, and let $\I$ contain those subsets of $U$ that are linearly independent over $\fieldF$. Then $(U,\I)$ defines a matroid $M$ and we say that $A$ \emph{represents} $M$. A matroid $M$ is \emph{representable (over \fieldF)} if there is a matrix $A$ (over \fieldF) that represents it. A matroid representable over at least one field is called \emph{linear}.

Let $D=(V,E)$ be a directed graph and $S,T\subseteq V$. The set $T$ is \emph{linked to $S$} if there exist $|T|$ vertex-disjoint paths from $S$ to $T$; paths of length zero are permitted. For a directed graph $D=(V,E)$ and $S,T\subseteq V$ the pair $M=(V,\I)$ is a matroid, where $\I$ contains those subsets $T'\subseteq T$ that are linked to $S$ \cite{Perfect1968}. Matroids that can be defined in this way are called \emph{gammoids}; the special case with $T=V$ is called a \emph{strict gammoid}. Marx~\cite{Marx09} gave an efficient randomized algorithm for finding a representation of a strict gammoid given the underlying graph; the error probability can be made exponentially small in the runtime.

\begin{theorem}[\cite{Perfect1968,Marx09}]\label{theorem:gammoidrepresentation}
Let $D=(V,E)$ be a directed graph and let $S\subseteq V$. The subsets $T\subseteq V$ that are linked to $S$ from the independent sets of a matroid over $M$. Furthermore, a representation of this matroid can be obtained in randomized polynomial time with one-side error.
\end{theorem}

As in previous work~\cite{KratschW12} we use the notion of \emph{representative sets}. The definition was introduced by Marx~\cite{Marx09} inspired by earlier work of Lov\'asz~\cite{Lovasz1977}.

\begin{definition}[\cite{Marx09}]\label{definition:representativesets}
Let $M=(U,\I)$ be a matroid and let $\Y$ be a family of subsets of $U$. A subset $\Y^*\subseteq\Y$ is \emph{$r$-representative} for $\Y$ if the following holds: For every $X\subseteq U$ of size at most $r$, if there is a set $Y\in\Y$ such that $X\cap Y=\emptyset$ and $X\cup Y\in\I$ then there is a set $Y^*\in\Y^*$ such that $X\cap Y^*=\emptyset$ and $X\cup Y^*\in\I$.
\end{definition}

Note that in the definition we may as well require that $\Y$ is a family of independent sets of $M$; independence of $X\cup Y$ requires independence of $Y$.

Marx~\cite{Marx09} proved an upper bound on the required size of representative subsets of a family $\Y$ in terms of the rank of underlying matroid and the size of the largest set in $U$. The upper bound proof is similar to \cite[Theorem 4.8]{Lovasz1977} by Lov\'asz.

\begin{lemma}[\cite{Marx09}]\label{lemma:representativeset}
Let $M$ be a linear matroid of rank $r+s$ and let $\Y=\{Y_1,\ldots,Y_m\}$ be a collection of independent sets, each of size of $s$. If $\Y>\binom{r+s}{s}$ then there is a set $Y\in\Y$ such that $\Y\setminus \{Y\}$ is $r$-representative for $\Y$. Furthermore, given a representation $A$ of $M$, we can find such a set $Y$ in $f(r,s)\cdot(||A||m)^{\Oh(1)}$ time.
\end{lemma}

The factor of $f(r,s)$ in the runtime of Lemma~\ref{lemma:representativeset} is due to performing linear algebra operations on vectors of dimension $\binom{r+s}{s}$. Since our application of the lemma has $s=3$ and $r$ bounded by the number of vertices in the underlying graph, this factor of the runtime is polynomial in the input size. We also remark that we will tacitly use the lemma for directly computing an $r$-representative subset $\Y^*\subseteq\Y$ of size at most $\binom{r+s}{s}$ since the lemma can clearly be iterated to achieve this. We note that faster algorithms for computing independent sets was given by Fomin et al.~\cite{FominLS14}, which leads to significantly better runtimes, in particular for the case of uniform matroids, when $s$ is not constant.

Now we are ready to prove Lemma~\ref{lemma:repsetofcriticalsets}. The proof follows the strategy used for \cite[Lemma 2]{KratschW11_arxiv}. For convenience, let us recall the lemma statement.

\begin{lemma}[recalling Lemma~\ref{lemma:repsetofcriticalsets}]
Let $H=(V_H,E_H)$ be a directed graph, let $S_H\subseteq V_H$, let $\ell\in\N$, and let $\T$ be a family of nonempty vertex sets $T\subseteq V_H$ each of size at most three. In randomized polynomial time, with failure probability exponentially small in the input size, we can find a set $\T^*\subseteq\T$ of size $\Oh(\ell^3)$ such that for any set $X_H\subseteq V_H$ of size at most $\ell$ that is closest to $S_H$ if there is a set $T\in\T$ such that all vertices $v\in T$ are reachable from $S_H$ in $H-X_H$ then there is a corresponding set $T^*\in\T^*$ satisfying the same properties.
\end{lemma}

\begin{proof}
We begin with constructing a directed graph $D$ and vertex set $S$:
\begin{enumerate}
 \item Create a graph $\overline{H}$ from $H$ by adding $\ell+1$ new vertices $s_1,\ldots,s_{\ell+1}$ and adding all edges $(s_i,s)$ for $i\in\{1,\ldots,\ell+1\}$ and $s\in S$. Define $\overline{S}:=\{s_1,\ldots,s_{\ell+1}\}$ and note that $V(\overline{H})=V_H\cup \overline{S}$.
 \item Let $D$ consist of three vertex-disjoint copies of $\overline{H}$. The vertex set $V^j$ of copy $j$ is $V^j=\{v^j\mid v\in V(\overline{H})$; let $S^j:=\{s_i^j\mid s_i \in \overline{S}\}\subseteq V^i$.
 \item Let $S:=S^1\cup S^2\cup S^3$. Note that $|S|=3(\ell+1)$.
\end{enumerate}
Let $M$ be the strict gammoid defined by graph $D$ and source set $S$. Compute in randomized polynomial time a matrix $A$ that represents $M$ using Theorem~\ref{theorem:gammoidrepresentation}; it suffices to prove that we arrive at the claimed set $\T^*$ if $A$ does indeed represent $M$, i.e., if no error occurred.

We now define a family $\Y$ of subsets of $V(D)$, each of size three; for convenience, let $<$ be an arbitrary linear ordering of the vertex set $V_H$ of $H$:
\begin{enumerate}
 \item For $\{u,v,w\}\in\T$ with $u<v<w$ let $Y(\{u,v,w\}):=\{u^1,v^2,w^3\}$.
 \item For $\{u,v\}\in\T$ with $u<v$ let $Y(\{u,v\}):=\{u^1,v^2,v^3\}$.
 \item For $\{u\}\in\T$ let $Y(\{u\}):=\{u^1,u^2,u^3\}$.
\end{enumerate}
Let us remark that the particular assignment of vertices in $T\in\T$ to the three disjoint copies of $\overline{H}$ is immaterial so long as copies of all vertices are present. The following claim relates reachability of vertices in $T\in\T$ in $H-X_H$ to independence of $Y(T)$ in $M$.

\begin{claim}\label{claim:reachability:independence}
Let $X_H\subseteq V_H$ be a set of at most $\ell$ vertices that is closest to $S_H$ in $H$, and let $T\in\T$. The vertices in $T$ are all reachable from $S_H$ in $H-X_H$ if and only if $Y(T)\cup I$ is independent in $M$ and $Y(T)\cap I=\emptyset$, where $I:=\{x^1,x^2,x^3\mid x\in X_H\}$.
\end{claim}

\begin{proof}
Assume first that each vertex of $T$ is reachable from $S_H$ in $H-X_H$. Observe that this requires $T\cap X_H=\emptyset$. By Proposition~\ref{proposition:closest}, since $X_H$ is closest to $S_H$, we have that there exist $|X_H|+1$ vertex-disjoint paths from $S_H$ to $X_H\cup\{v\}$ for each vertex $v\in T$; in other words, $X_H\cup\{v\}$ is linked to $S_H$ in $H$. Since $|X_H\cup\{v\}|\leq \ell$, it follows directly that $X_H\cup\{v\}$ is linked to $\overline{S}$ in $\overline{H}$. Thus, for $v^j\in Y(T)$ with $j\in\{1,2,3\}$, it follows that $I^j\cup\{v_j\}$ is linked to $S^j$ in $D$, where $I^j:=\{x^j \mid x\in X_H\}$. Since the three copies of $\overline{H}$ in $D$ a vertex-disjoint, we conclude that $Y(T)\cup I^1\cup I^2\cup I^3=Y(T)\cup I$ is linked to $S=S^1\cup S^2\cup S^3$ in $D$. Thus, $Y(T)\cup I$ is independent in $M$, as claimed. To see that $Y(T)\cap I=\emptyset$ note that $v^j\in Y(T)\cap I$ would imply $v\in T$ and $v\in X_H$; a contradiction to $T\cap X_H=\emptyset$.

For the converse, assume that $Y(T)\cup I$ is independent in $M$ and that $Y(T)\cap I=\emptyset$. Let $v\in T$ and let $j\in\{1,2,3\}$ such that $v^j\in Y(T)$. Observe that $Y(T)\cap I=\emptyset$ implies $v\notin X_H$: Indeed, if $v\in X_H$ then we have $v^j$ in $Y(T)$; a contradiction. Now, independence of $Y(T)\cup I$ implies that $Y(T)\cup I$ is linked to $S$ in $D$. It follows, by vertex-disjointness of the three copies of $\overline{H}$ in $D$ that $y^j\cup I^j$ is linked to $S^j$ using only vertices $v^j$ with $v\in V(\overline{H})$. This implies that $y\cup X_H$ is linked to $\overline{S}$ in $\overline{H}$. Observe now that any path from $\overline{S}$ to $y\cup X_H$ must contain as its second vertex a vertex of $S$; here it is convenient that $X_H\subseteq V_H$ and $\overline{S}\cap V_H=\emptyset$, causing all paths to have length at least one and at least two vertices. Thus, we conclude that $y\cup X_H$ is linked to $S$ in $\overline{H}$, and hence also in $H$ since vertices of $\overline{S}$ cannot be internal vertices of paths (as they have only outgoing edges). Clearly, in a collection of $|X_H|+1$ paths from $S$ to $X_H\cup \{v\}$ the path from $S$ to $v$ cannot contain any vertex of $X_H$ as they are endpoints of the other paths. Thus, there exists a path from $S$ to $v$ that avoids $X_H$, implying that $v$ is reachable from $S$ in $H-X_H$. Since $v$ was chosen arbitrarily from $T$, the claim follows.
\end{proof}

Now, use Lemma~\ref{lemma:representativeset} on the gammoid $M$ defined by graph $D$ and source set $S$, represented by the matrix $A$. The rank of $M$ is obviously exactly $|S|=3\ell+3$ since no set larger than $S$ can be linked to $S$ and $S$ itself is an independent set (as it is linked to itself). For the lemma choose $r=|S|-3=3\ell$ and $s=3$ and note that all sets in $\Y$ have size exactly $s=3$ as required. We obtain a set $\Y^*$ of size at most $\binom{r+s}{s}=\Oh(|S|^3)=\Oh(\ell^3)$ that $r$-represents $\Y$. Define a set $\T^*\subseteq\T$ by letting $\T^*$ contain those sets $T\in\T$ with $Y(T)\in\Y^*$. The size of $\T^*$ is equal to $|\Y^*|=\Oh(\ell^3)$ since each $Y\in\Y^*$ has exactly one $T\in\T$ with $Y=Y(T)$. (To see this, note that dropping the superscripts in $Y$ yields exactly the members of the corresponding set $T$; some may be repeated.)

\begin{claim}\label{claim:tstar}
For any set $X_H\subseteq V_H$ of size at most $\ell$ that is closest to $S_H$ if there is a set $T\in\T$ such that all vertices $v\in T$ are reachable from $S_H$ in $H-X_H$ then there is a corresponding set $T^*\in\T^*$ satisfying the same properties.
\end{claim}

\begin{proof}
Let $T\in\T$ such that all vertices $v\in T$ are reachable from $S_H$ in $H-X_H$. By Claim~\ref{claim:reachability:independence} the set $Y(T)\cup I$ is independent in $M$ and $Y(T)\cap I=\emptyset$, where $I=\{x^1,x^2,x^3\mid x\in X_H\}$. Note that $Y(T)$ in $\Y$ and that $|I|=3|X_H|\leq 3\ell=r$. Thus, by Lemma~\ref{lemma:representativeset} there must be a set $Y^*\in\Y^*$ such that $Y^*\cap I=\emptyset$ and $Y^*\cup I$ is an independent set of $M$. Let $T^*\in\T$ with $Y^*=Y(T^*)$; such a set $T^*$ exists by definition of $\Y$ and, as discussed above, it is uniquely defined. By Claim~\ref{claim:reachability:independence} it follows that that all vertices of $T$ are reachable from $S_H$ in $H-X_H$. This completes the proof of Claim~\ref{claim:tstar}.
\end{proof}

We recall that for our case of $s=3$ the computation of $\Y^*$ can be seen to take time polynomial in the input size. The computed set $\Y^*$ fulfills the lemma statement unless the gammoid representation computed by Theorem~\ref{theorem:gammoidrepresentation} is erroneous, which has exponentially small chance of occurring. Note that boosting the success chance of Theorem~\ref{theorem:gammoidrepresentation} works by increasing the range of the random integers used therein (respectively, the field size): An additional factor of $2^p$ in the range of integers decreases the error probability by a factor of $2^{-p}$, while increasing the encoding size of the integers only by $p$ bits. Thus, by only a polynomial increase in the running time, we can get exponentially small error. This completes the proof.
\end{proof}


\section{Conclusion}\label{section:conclusion}

We have presented a randomized polynomial kernelization for \vcanb by giving a (randomized) polynomial parameter transformation to \vcamm. This improves upon the smallest parameter, namely $k-LP(G)$, for which such a result was known~\cite{KratschW12}. The kernelization for \vcamm \cite{KratschW12} involves reductions to and from \atwosat, which can be done without affecting the parameter value (cf.~\cite{RamanRS11}). We have not attempted to optimize the total size. Given an instance $(G,k,\ell)$ for \vcanb we get an equivalent instance of \atwosat with $\Oh(k^{24})$ variables and size $\Oh(k^{48})$, which still needs to be reduced to a \vc instance. 

It seems likely that the kernelization can be improved if one avoids the blackbox use of the kernelization for \vcamm and the detour via \atwosat. In particular, the underlying kernelization for \atwosat applies, in part, the same representative set machinery to reduce the number of a certain type of clauses. Conceivably the two applications can be merged, thus avoiding the double blow-up in size. As a caveat, it appears to be likely that this would require a much more obscure translation into a directed separation problem. Moreover, the kernelization for \atwosat requires an approximate solution, and it is likely that the same would be true for this approach. It would of course also be interesting whether a deterministic polynomial kernelization is possible, but this is, e.g., already not known for \atwosat and \vcamm.

We find the appearance of a notion of critical sets of size at most three and the derived separation problem in the auxiliary directed graph quite curious. For the related problem of separating at least one vertex from each of a given set of triples from some source $s$ by deleting at most $\ell$ vertices (a variant of \problem{Digraph Paircut}~\cite{KratschW12}) there is a natural $\Oh^*(3^\ell)$ time algorithm that performs at most $\ell$ three-way branchings before finding a solution (if possible). It would be interesting whether a complete encoding of \vcanb into a similar form would be possible, since that would imply an algorithm that exactly matches the running time of the algorithm of the algorithm by Garg and Philip~\cite{GargP16}.



\end{document}